\documentclass[a4paper,11pt]{scrartcl}

\usepackage[T1]{fontenc}
\usepackage{appendix}
\usepackage{graphicx} %
\usepackage{amsmath}
\usepackage{amssymb}
\usepackage{amsthm}

\usepackage{authblk}
\usepackage[english]{babel}
\usepackage{hyperref}
\usepackage{xcolor}
\hypersetup{
    colorlinks,
    linkcolor={red!50!black},
    citecolor={blue!50!black},
    urlcolor={blue!80!black}
}
\usepackage{mathtools}
\mathtoolsset{showonlyrefs,showmanualtags}

\usepackage{soul}
\usepackage{todonotes}
\newtheorem{thm}{Theorem}[section]
\newtheorem{defn}[thm]{Definition}
\newtheorem{prop}[thm]{Proposition}
\newtheorem{corollary}[thm]{Corollary}
\newtheorem{remark}[thm]{Remark}
\newtheorem{lemma}[thm]{Lemma}
\newtheorem{example}[thm]{Example}
\usepackage[autostyle=true]{csquotes}
\usepackage[backend=bibtex,natbib=false,style=numeric,sortcites=true]{biblatex} %
\renewbibmacro{in:}{}
\bibliography{hdec.bib}

\newcommand{\red}[1]{#1}
\newcommand{\cH}{\mathcal{H}}
\newcommand{\Reeb}{\mathcal{R}}
\usepackage{todonotes}
\title{Characterization of symmetries of contact Hamiltonian systems}

\author[1]{Federico Zadra}
\author[2]{Marcello Seri}
\affil[1]{Department of Mathematics, University of Antwerp, Belgium}
\affil[2]{Bernoulli Institute for Mathematics, Computer Science and Artificial Intelligence, University of Groningen, The Netherlands}
\date{}

\begin{document}
\maketitle

\begin{abstract}
	This paper explores the relationship between Cartan symmetries, dynamical similarities, and dynamical symmetries in contact Hamiltonian mechanics.
	By introducing an alternative decomposition of vector fields, we characterize these symmetries and present a novel description in terms of tensor densities.
	Furthermore, we demonstrate that this framework allows, under specific conditions, for the recovery of integrals of motion.
	We also establish new criteria to assess their independence.

	\medskip
	\noindent
	\textbf{Keywords:} contact geometry, dynamical symmetries, Cartan symmetries, scaling symmetries, dynamical similarities, Hamiltonian systems

	\noindent
	\textbf{MSC2010:} 70G45, 70G65, 70H33, 34A26, 53D10
\end{abstract}

\section{Introduction}

Contact Hamiltonian systems are a natural generalization of the usual symplectic Hamiltonian systems to odd-dimensional manifolds~\cite{Arnold2005,Arnold2009}.
The move to odd-dimensional spaces allows for the description of a wider range of physical systems, most notably those with dissipation, leading to multiple efforts towards a comprehensive contact geometric description of such dynamics~\cite{Bravetti2015,Bravetti2017,Leon2019}.

Despite the differences that distinguish genuinely Hamiltonian systems from their contact counterparts, it is possible to extend many results across the two settings, including analogues to Noether's theorem~\cite{Bravetti2021, Bravetti2023a} and notions of integrability~\cite{Leon2019}.
For contact systems any $1$-parameter transformation group that preserves the contact structure is a contact Hamiltonian action~\cite{geiges}.
However, in this setting the conservation of the distribution does not imply a conservation of quantities along the motion, and the definition and characterization of symmetries and invariants becomes much more involved.

In fact, in contact systems there is no unique natural way to define the notion of a symmetry, and in the past researchers have already explored different definitions and their properties~\cite{Bravetti2021,Bravetti2023a,Bravetti2017,Leon2019}: Cartan symmetries, dynamical symmetries and dynamical similarities.

\bigskip
In this paper, we explore the relationship between these notions of symmetry.
We provide a new characterization of such infinitesimal symmetries by an alternative decomposition of vector fields, which we refer to as the Hamiltonian--horizontal decomposition~\cite{projective, vectorcontactstructure, Zadra}.
To ensure our construction is independent of the choice of contact form, we introduce an alternative description in terms of tensor densities, which provides an additional, intrinsic tool for the characterization.

This framework allows us to systematically compare and relate the different symmetry classes.
In Section~\ref{sec:decsym}, we derive the main results of our work: a series of characterizations that provide necessary and sufficient conditions for a vector field to belong to a given class of symmetry.
In particular, we show that a vector field is a dynamical symmetry if and only if its Hamiltonian component is a dissipated quantity for the system (Theorem~\ref{thm:dynsym}).
We also clarify the geometric origin of the auxiliary function in the definition of Cartan symmetries (Theorem~\ref{thm:cartandeco} and discussion thereafter) and, for the important case of scaling symmetries, we provide a method to recover conserved quantities (Theorem~\ref{thm:scalingconstant} and discussion thereafter).

The effectiveness of these characterizations is demonstrated on concrete mechanical contact systems in Section~\ref{sec:applications}.
Finally, in Section~\ref{sec:integra}, we explore how to study the integrability of contact systems and the role of this new characterization of symmetries in recovering integrals of motion.
For instance, we give a practical recipe to check if a set of dissipated quantities are independent (Theorem~\ref{thm:independence}) and thus can be used to prove the contact integrability of a system and we show how scaling symmetries can be systematically used to generate new functions in involution and how to check their independence (Theorem~\ref{thm:scalingindependence}).

\bigskip
In more detail, the paper is structured as follows.
In Section~\ref{sec:contactgeometry} we briefly review contact Hamiltonian systems and introduce the necessary notation.
In Section \ref{sec:Decomposition} we introduce and describe the horizontal--Hamiltonian decomposition and its expression in terms of tensor densities.
In Section \ref{sec:decsym} we derive the characterization of symmetries and dissipated quantities of contact Hamiltonian systems. All the notions of symmetry are discussed in their separate subsections: Dynamical symmetries in Section~\ref{sec:dynsym}, Dynamical similarities in Section~\ref{sec:dynsimil}, with a particular focus on the special case of Scaling symmetries in Section~\ref{sec:scalsymm}, and Cartan symmetries in Section~\ref{sec:cartan}.
In Section~\ref{sec:applications}, we illustrate the use of these characterizations on some relevant examples from contact mechanics.
Finally, in Section~\ref{sec:integra}, we begin to explore how this new characterization of symmetries can help to reconstruct constants of motion for contact Hamiltonian systems.

\section{Contact Geometry}
\label{sec:contactgeometry}
In this section, we will introduce the basic concepts of contact geometry, focusing on the aspects that are relevant for our analysis of contact Hamiltonian systems.

\begin{defn}
	\label{def:contactmanifold}
	A \emph{contact manifold} $(M,\Delta)$ is a pair of a differentiable manifold $M$ of odd dimension $2n+1$ and a maximally non-integrable smooth distribution $\Delta \subset TM$ of co-dimension $1$.
	We say that the contact manifold is \emph{exact} (or co-oriented~\cite{geiges}) if the distribution $\Delta$ coincides with the kernel of a $1$-form $\eta$ such that
	\begin{align}
		\eta \wedge (d \eta)^{\wedge n} \neq 0, \label{eqns:contactvol}
	\end{align}
	where \begin{equation}
		d\eta^{\wedge n} = \underbrace{d\eta \wedge \ldots \wedge d\eta}_{n-\text{times}}.
	\end{equation}
	We will call the distribution $\Delta$ the \emph{horizontal distribution}, and its sections \emph{horizontal vector fields}.
\end{defn}

\begin{remark}
	With an abuse of notation, we will refer to $\Delta$ also as the space of the sections $\Gamma(\Delta)$ over the manifold $M$.
\end{remark}

The distinguishing property of exact contact manifold is that we can globally describe the contact distribution by just considering one particular $1$-form, what is also referred to as the co-orientation of the contact distribution.
This is not generally true: a local contact form always exists but it may not possible to extend it globally.
In any case, including exact contact manifolds, the choice of the contact form is not unique: any other $1$-form that is proportional to the contact form by a non-vanishing real function is again a contact form that describes the same contact distribution, i.e. for all $f:M \to \mathbb{R}^+$
\begin{equation}\label{eqn:changecontactform}
	\ker \eta = \ker f \eta.
\end{equation}

Once a contact form has been fixed, the $(2n+1)$-form in equation~\eqref{eqns:contactvol} is the natural choice of volume form for the contact manifold.
Moreover, we can use the non-integrability condition~\eqref{eqns:contactvol} to describe the part of the tangent space that is transverse to the contact distribution.
This transversal component, leads to the choice of a particular generator for the complement of the contact distribution: the Reeb vector field.

\begin{prop}[\cite{geiges}]
	\label{prop:reeb}
	Let $(M,\eta)$ an exact contact manifold. There exists a unique vector field $\Reeb\equiv\Reeb_\eta\in\Gamma(TM)$ such that
	\begin{equation}
		\eta(\Reeb) = 1, \quad\mbox{and}\quad
		d\eta(\Reeb,\cdot) = 0.
	\end{equation}
	Such a vector field is called \emph{Reeb vector field} for $\eta$.
\end{prop}

It should be clear from the previous statement that the Reeb vector field $\Reeb_\eta$ is tightly connected to the choice contact form $\eta$, and in general will not be as rigid as $\ker \eta$. Indeed, despite the simplicity of equation~\eqref{eqn:changecontactform}, the corresponding transformation law $\Reeb_\eta \mapsto \Reeb_{f\eta}$ for the Reeb vector field is not analogously nice and the flows of the two Reeb vector fields can be quite different.

To study in depth the transformation of contact forms and related quantities, we should introduce the notion of contactomorphism, that is, a transformation that preserves the contact distribution.

\begin{defn}
	Let $(M,\Delta_M)$ and $(N,\Delta_N)$ be two contact manifolds. A diffeomorphism $\phi:M\to N$ is called a \emph{contactomorphism} (or contact transformation) if
	\begin{equation}
		\label{eqn:contactomorphismTM}
		\Delta_M \subseteq \phi^*(\Delta_N).
	\end{equation}
\end{defn}
This definition can be extended to exact contact manifolds by rewriting the condition~\eqref{eqn:contactomorphismTM} in terms of the contact form.

\begin{prop}
	\label{prop:contactomorphism}
	Let $(M,\eta_M)$ and $(N,\eta_N)$ be two exact contact manifolds. If $\phi : M \to N$ is a smooth contactomorphism, then
	\begin{equation}
		\phi^*\eta_N = F_\phi \eta_M,
	\end{equation}
	where $F_\phi\in C^\infty(M)$ is a non-vanishing smooth real function.
	Moreover, if $\phi^t : \mathbb{R}\times M \to N$,
	is a $1$-parameter group of smooth contactomorphisms, its infinitesimal generator $Y_\phi$ satisfies
	\begin{equation}
		L_{Y_\phi} \eta = - f_{Y_\phi} \eta,
	\end{equation}
	where $f_{Y_\phi} = \frac{d F_{\phi^t}}{d t} \in C^\infty(M)$ is a smooth real function.
	We refer to $Y_\phi$ as \emph{infinitesimal contactomorphism}.
\end{prop}

\begin{remark}
	\label{rmk:volume}
	Let $(M,\eta)$ be an exact contact manifold. If a contactomorphism $g: M \to M$ acts on $\eta$ as
	\begin{equation}
		\label{eqn:transfcontactform}
		g^* \eta = f_g \eta,
	\end{equation}
	by a non-vanishing function $f_g$, the volume form $\eta \wedge (d\eta)^{\wedge n}$ is transformed as
	\begin{equation}
		\label{eqn:transfvolume}
		g^*( \eta \wedge (d\eta)^{\wedge n} )= (f_g)^{n+1} \eta \wedge (d\eta)^{\wedge n}.
	\end{equation}
\end{remark}

\subsection{Contact Hamiltonian systems}
\label{subsec:chs}
Any 1-parameter group of contactomorphisms is the flow of a contact Hamiltonian function~\cite{geiges}, so we have already implicitly introduced contact Hamiltonian systems in Proposition~\ref{prop:contactomorphism}.
A more practical definition can be obtained by mimicking the construction of Hamiltonian vector fields in the symplectic setting~\cite{Bravetti2017, Leon2019}.

\begin{defn}
	A \emph{contact Hamiltonian system} is a triple $(M,\eta,\cH)$ where the pair $(M,\eta)$ is an exact contact manifold, and $\cH \in C^{\infty}(M)$ is a real function called \emph{contact Hamiltonian}. Any contact Hamiltonian induces a vector field, called the \emph{contact Hamiltonian vector field} associated to $\cH$, defined by the equations
	\begin{equation}
		\label{eqn:Hamvec}
		\eta(X_\cH) = - \cH, \qquad d\eta(X_\cH, \cdot) = d\cH - \Reeb(\cH) \eta.
	\end{equation}
\end{defn}

One of the main differences between the symplectic and the contact settings is that in the latter, the Hamiltonian flow does not preserve the Hamiltonian function. Indeed, it follows immediately from~\eqref{eqn:Hamvec} that
\begin{equation}
	\label{eqn:nonconservation}
	X_\cH (\cH) = - \Reeb(\cH) \cH.
\end{equation}
For a symplectic Hamiltonian system, the left-hand side of \eqref{eqn:nonconservation} vanishes and, therefore, the energy (i.e., in this case, the Hamiltonian itself) is conserved along the Hamiltonian flow, see~\eqref{eq:hconsen}.
Therefore, in the contact case, unless $\cH = 0$ or $\Reeb(\cH) = 0$, one can think of \eqref{eqn:nonconservation} as encoding a ``non-conservation'' of energy.

\begin{remark}
	Equation~\eqref{eqn:Hamvec} and Cartan's magic formula, imply that
	\begin{equation}\label{eq:LieHameta}
		L_{X_\cH} \eta = d \iota_{X_\cH} \eta + \iota_{X_\cH} d \eta = - \Reeb (\cH) \eta \,.
	\end{equation}
	This is a very convenient identity when working with contact Hamiltonian vector fields, and is sometimes used as an alternative definition (together with $\eta(X_\cH)=-\cH$), in place of the right hand side of~\eqref{eqn:Hamvec}.
\end{remark}
The two propositions that follow are important tools for studying infinitesimal symmetries of contact Hamiltonian systems. As pointed out already in Proposition~\ref{prop:contactomorphism}, the action of an infinitesimal contactomorphism leaves the distribution invariant, and contact Hamiltonian vector fields are infinitesimal contactomorphisms.
In fact, the converse is also true~\cite{geiges,Leon2019}: any infinitesimal contactomorphism is a contact Hamiltonian vector field in the sense of~\eqref{eqn:Hamvec}.

\begin{prop}
	\label{prop:automorphismhor}
	Contact Hamiltonian vector fields preserve the contact distribution, i.e. if $X_f$ is one such vector field, $[X_f, \Delta] \subseteq \Delta$.
\end{prop}

\begin{proof}
	Consider a contact Hamiltonian vector field $X_g$ and
	let $Y \in \Delta$ be an arbitrary vector field tangent to the contact distribution.
	Computing the Lie bracket, and observing that $\eta(Y)=0$, we get
	\begin{align}
		\eta(\left[X_g, Y \right]) & = d\eta(X_g,Y) + Y(\eta(X_g)) - X_g(\eta(Y)) \\
		                           & = dg(Y) - R(g) \eta(Y) - dg(Y)               \\
		                           & = 0,
	\end{align}
	therefore $[X_g,Y] \in \Delta$.
\end{proof}

Moreover, the space of contact Hamiltonian vector fields is closed under the action of Lie brackets, as shown in the next proposition.

\begin{prop}
	\label{prop:preservetionHam}
	If $X_g$ and $X_f$ are contact Hamiltonian vector fields, then their Lie bracket $[X_g,X_f] = X_{-\eta([X_g,X_f])}$ is itself a contact Hamiltonian vector field with Hamiltonian function $-\eta([X_g,X_f])$.
\end{prop}

\begin{proof}
	Let $h := - \eta([X_g, X_f])$.
	We can compute the Lie derivative $L_{\left[X_g,X_f\right]} \eta$ in two ways by using the definition of a contact Hamiltonian vector field and equation~\eqref{eq:LieHameta}:
	\begin{align}
		 & L_{\left[X_g,X_f\right]} \eta = - d h + d\eta (\left[X_g,X_f\right], \cdot),
		\qquad\mbox{and}                                                                                                                \\
		 & L_{\left[X_g,X_f\right]} \eta = L_{X_g}L_{X_f}\eta - L_{X_f}L_{X_g} \eta= - \left(X_g(\Reeb(f)) - X_f(\Reeb(g))\right) \eta.
	\end{align}
	Thus, in particular, we have
	\begin{equation}
		\label{eqn:preservationone}
		- \left(X_g(\Reeb(f)) - X_f(\Reeb(g))\right) \eta + d h - d\eta (\left[X_g,X_f\right], \cdot) = 0.
	\end{equation}
	Contracting \eqref{eqn:preservationone} with the Reeb vector field, we obtain
	\begin{equation}
		-X_g(\Reeb(f)) + X_f(\Reeb(g)) + \Reeb(h) = 0.
	\end{equation}
	This can be used in~\eqref{eqn:preservationone} to get
	\begin{equation}
		d\eta (\left[X_g,X_f\right], \cdot) = d h - \Reeb(h) \eta.
	\end{equation}

	In summary, if $Y := [X_g, X_f]$, we have shown that
	\begin{equation}\label{eqn:contact-liebracketham}
		\eta(Y) = - h, \qquad\mbox{and}\qquad d\eta(Y, \cdot) = d h - \Reeb(h) \eta,
	\end{equation}
	that is, $[X_g, X_f] = X_h$ is the contact Hamiltonian vector field associated to $h$.
\end{proof}

In what follows, we also need to introduce a particular contact Hamiltonian system, whose contact Hamiltonian is the constant function $\cH \equiv - 1$.
\begin{defn}
	The \emph{Reeb system} on a contact manifold $(M,\eta)$ is the contact Hamiltonian system defined by the triple $(M,\eta,\cH \equiv -1)$.
	In the rest of the manuscript we will abuse notation and denote it simply $(M,\eta,-1)$.
\end{defn}

At least locally, away from its zero level set, every contact Hamiltonian system can be straightened into a Reeb system~\cite{Bravetti2015}.
Indeed, if the Hamiltonian of a contact Hamiltonian system $(M,\eta,\cH)$ is not vanishing at a point $x_0\in M$, then there is an open neighbourhood $V\subset M$ of $x_0$ on which $\cH\neq 0$. In this neighbourhood we can define a contact transformation that reduces the contact Hamiltonian system to a Reeb system $(V,\Tilde{\eta},-1)$:
the $1$-form
\begin{equation}
	\tilde{\eta}\rvert_V := - \frac{1}{\cH} \eta\rvert_V,
\end{equation}
defines on $V$ the same contact structure $\Delta = \ker \Tilde{\eta} = \ker \eta$ as $\eta$.
Moreover, the Hamiltonian vector field $X_\cH$ is the Reeb vector field for $\Tilde{\eta}$: this directly follows from the explicit computation
\begin{equation}
	\Tilde{\eta}(X_\cH) = -\frac{\eta(X_\cH)}{\cH} = 1, \qquad
	d \Tilde{\eta}(X_\cH,\cdot) = \left(+\frac{d\cH\wedge\eta}{\cH^2} - \frac{d\eta}{\cH} \right)(X_\cH, \cdot) = 0,
\end{equation}
and by uniqueness of the Reeb vector field, see Proposition~\ref{prop:reeb}.
\begin{remark}
	\label{rmk:reebHam}
	On the converse, a contact transformation $\psi:(M,\eta) \to (N,\eta')$ whose pullback scales the contact form by a non-vanishing function $F_\psi:M\to \mathbb{R}_+$ as in Proposition~\ref{prop:contactomorphism}, maps the Reeb vector field into a Hamiltonian vector field~\cite{geiges,Zadra}:
	\begin{equation}
		\psi_* \Reeb = - X'_{F_\psi \circ \psi^{-1}},
	\end{equation}
	where for $X'$ we refer to the contact Hamiltonian vector field computed from $\eta'$.
\end{remark}

\subsection{Jacobi Structure on contact manifolds}
\label{sec:Jacobi}
The function $h = -\eta([X_g,X_f])$ introduced in the proof of Proposition~\ref{prop:preservetionHam} has a deeper meaning in the context of contact Hamiltonian systems. It represents the natural Jacobi bracket structure on the functions on a contact manifold~\cite{Leon2019,de_Le_n_2020}.

\begin{defn}[\cite{Kirillov1976}]
	A \emph{Jacobi manifold} is a smooth manifold $M$ endowed with a \emph{Jacobi structure}, that is, a pair of a skew bivector field (i.e., a skew-symmetric, contravariant 2-tensor on $TM$) $\Lambda(\cdot, \cdot)$ and a vector field $\mathcal{E} \in \Gamma(TM)$ such that
	\begin{align}
		[\Lambda, \mathcal{E}]_{SN} & = 0,                             \\
		[\Lambda, \Lambda]_{SN}     & = -2 \Lambda \wedge \mathcal{E},
	\end{align}
	where $[\cdot, \cdot]_{SN}$ denotes the Schouten-Nijenhuis brackets~\cite{Marsden2010}.
\end{defn}

A Jacobi structure on a manifold induces a bilinear map between functions that are continuously differentiable; we call this map \emph{Jacobi bracket}.
\begin{defn}
	If $(M,\Lambda,\mathcal{E})$ is a Jacobi manifold, there is a natural product structure on smooth functions called \emph{Jacobi bracket} and defined by
	\begin{equation}
		\{f,g \}_J = \Lambda(df, dg) - g \mathcal{E}(f) + f \mathcal{E}(g).
	\end{equation}
\end{defn}
\begin{remark}
	If $\mathcal{E} = 0$, then the structure induced by $\Lambda(\cdot, \cdot)$ is a Poisson structure~\cite{Kirillov1976}.
\end{remark}
\begin{remark}
	The bivector $\Lambda$ induces a morphism of vector bundles $T^*M\to T^{**}M \simeq TM$ via the mapping $\nu \mapsto \Lambda(\nu, \cdot)$. This will be very handy in the constructions that follow.
\end{remark}

Let us revisit this in the contact setting. The choice of a contact form on an exact contact manifold, induces a Jacobi structure on the manifold and thus a Jacobi bracket on smooth functions.

The non-integrability of the contact structure on an exact contact manifold $(M,\eta)$ allows one to define an isomorphism of vector bundles
\begin{align}
	\flat & : TM\to T^*M                                   \\
	      & \quad X \mapsto \iota_X\, d\eta +\eta(X) \eta,
\end{align}
and a corresponding isomorphism of $C^\infty(M)$-modules between vector fields and one-forms that we will denote by the same symbol~\cite[Section 2.2]{de_Le_n_2020}.
This then allows to construct a natural Jacobi structure on the manifold as follows.

\begin{lemma}[\cite{de_Le_n_2020}]
	\label{lemma:bivector}
	On a contact manifold $(M,\eta)$, there is a natural Jacobi structure defined by the contact form $\eta$ as follows:
	\begin{equation}
		\Lambda(\sigma, \nu) = -d\eta(\flat^{-1} \sigma, \flat^{-1} \nu), \quad \mbox{and}\quad \mathcal{E} = - \Reeb \mbox{ is the Reeb vector field}.
	\end{equation}
\end{lemma}

\begin{remark}
	On an exact contact manifold $(M, \eta)$, we denote by $\{\cdot, \cdot\}_\eta := \{\cdot, \cdot\}_J$ the Jacobi brackets induced by the contact form.
\end{remark}

Similarly to the symplectic case, the contact Hamiltonian vector field has an equivalent expression in terms of the Jacobi structure.

\begin{lemma}
	Let $(M,\eta)$ be a contact manifold and $(\Lambda, \mathcal{E})$ be the corresponding Jacobi structure. Then, the contact Hamiltonian vector field $X_f$ associated to $f$ has the form
	\begin{equation}\label{eqn:contactHam-Lambda}
		X_f = \Lambda(df, \cdot) + f \mathcal{E} = \Lambda(df, \cdot) - f \Reeb.
	\end{equation}
	Moreover, for any $f,g \in C^\infty(M)$, the Jacobi bracket $\{f,g\}_\eta$ can be rewritten in terms of the contact structure as
	\begin{equation}
		\label{eqn:commnot0}
		\{f,g\}_\eta = X_f(g) + g \Reeb(f).
	\end{equation}
\end{lemma}

Proposition~\ref{prop:preservetionHam}, then, implies that we can relate the Lie bracket of two contact Hamiltonian vector fields to the Jacobi brackets of their Hamiltonian functions: between the set of the contact Hamiltonian vector fields and the contact Hamiltonian functions there is an anti-isomorphism of Lie algebras.

\begin{prop}
	\label{prop:hamiltonianpropbrackets}
	Let $X_f$ and $X_g$ two Hamiltonian vector fields. Then,
	\begin{equation}
		\eta([X_f, X_g]) = -\{f, g\}_\eta, \label{eqn:JacobiHamiltonian}
	\end{equation}
	and
	\begin{equation}
		\left[X_f, X_g \right] = X_{\left\{f,g \right\}_\eta}. \label{eqn:LieHamiltonians}
	\end{equation}
	That is, the Lie bracket of these two vector fields is a contact Hamiltonian vector field with Hamiltonian function $\{f,g\}_\eta$.
\end{prop}
\begin{proof}
	The first equation is proven in \cite[Lemma 3]{de_Le_n_2020}. Then Proposition~\ref{prop:preservetionHam} implies
	\begin{equation}
		\left[X_f, X_g \right] =  X_{-\eta([X_f,X_g])} = X_{\left\{f,g \right\}_\eta}.
	\end{equation}
\end{proof}

\section{Decompositions of vector fields on contact manifolds}
\label{sec:Decomposition}
On an exact contact manifold $(M,\eta)$, the contact form $\eta$ naturally induces a splitting of the tangent space $T_x M$ at any point $x \in M$ into a component tangential to the horizontal distribution and the rest~\cite{geiges}.
This local horizontal--vertical decomposition can be naturally extended to vector fields over $M$ with the aid of the Reeb vector field, since in this case $TM = \mathrm{span}\{\Reeb, \Delta\}$ holds globally.

\begin{prop}
	Let $(M,\eta)$ a contact manifold,
	the $C^{\infty}(M)$ module of vector fields on $M$ admits the decomposition
	\begin{equation}
		\label{eqn:splittingTM}
		\Gamma(TM) = \ker \eta \oplus \ker d\eta.
	\end{equation}
\end{prop}
In view of equation~\eqref{eqn:splittingTM}, given $X \in \Gamma(TM)$, we can decompose it into
\begin{equation}
	X = h_X + \eta(X) \Reeb;
\end{equation}
where $h_X \in \Delta = \ker \eta$ and $\eta(X)$ is a smooth function. The vector field $\eta(X) \Reeb$ is called the \emph{vertical part} of $X$.

For our purposes, it will be convenient to consider a different decomposition, introduced in~\cite{projective,vectorcontactstructure,Zadra}, that replaces the vertical component with a contactomorphism.

\subsection{The Hamiltonian--horizontal decomposition}\label{sec:hamhor}
The Hamiltonian--horizontal decomposition splits vector fields on a contact manifold into a component described by a contact Hamiltonian vector field and a horizontal one.

\begin{prop}
	\label{prop:decompositionHami}
	Let $(M,\eta)$ be a contact manifold. Any vector field $\xi \in \Gamma(TM)$ admits a unique decomposition
	\begin{equation}
		\xi = X_{\phi_{\xi}} + \delta_\xi,
	\end{equation}
	where $\phi_\xi = \ -\eta(\xi) \in C^\infty(M)$ and $\delta_{\xi} \in \Delta = \ker \eta$.
\end{prop}
\begin{proof}
	Equation~\eqref{eqn:Hamvec} applied to $\delta_{\xi}=\xi - X_{\phi_\xi}$ implies
	\begin{equation}
		\eta(\xi - X_{\phi_\xi}) = \eta(\xi) - \eta(X_{-\eta(\xi)}) = 0.
	\end{equation}
	That is, $\delta_\xi$ a horizontal vector field.

	Since $\eta(\xi)$ is a well-defined smooth function on $M$, its Hamiltonian vector field $X_{\phi_\xi}$ is uniquely defined, and therefore the horizontal vector field $\delta_{\xi} = \xi - X_{\phi_\xi}$ is also unique.
\end{proof}

Note the different nature of these decompositions: the horizontal--vertical decomposition {can be} defined pointwise, that is, such a decomposition of $T_x M$ only involves the values of the vector fields at the point $x$; the Hamiltonian--horizontal decomposition, instead, is necessarily defined in a neighbourhood of the point as it also requires knowledge of the derivative of the vector field, and, consequently, can not be defined pointwisely.

As a matter of fact, by considering all the Hamiltonian vector fields evaluated at a point $x$ it is possible to span the tangent space $T_xM$, and then also the entire kernel of the contact form evaluated in the point $x$.
\begin{remark}
	Consider a point $x$ on a contact manifold of dimension $2n+1$. By Darboux's theorem, there exists a neighbourhood $U$ of $x$ in which the contact form can be written in the canonical form
	\begin{equation}
		\eta = ds - p_i \, dq^i, \quad i \in \{1, \ldots, n\}.
	\end{equation}
	In these coordinates, we may assume without loss of generality that
	\begin{equation}
		x = (0, \ldots, 0).
	\end{equation}
	At this point, the Hamiltonian vector fields $X_s$, $X_{q_i}$, and $X_{p_i}$ lie in $\ker \eta \lvert_x$, whereas the Reeb vector field $\Reeb$, which is always Hamiltonian, does not. Moreover, while $X_s(x) = 0$, we have
	\begin{equation}
		X_{q^i} = - \frac{\partial}{\partial p_i},
		\qquad
		X_{p_j} = \frac{\partial}{\partial q^j}.
	\end{equation}
	Therefore,
	\begin{equation}
		\operatorname{span}_{\mathbb{R}}\{ \Reeb(x), X_{q^i}, X_{p_j} \} = T_x M.
	\end{equation}
\end{remark}

{We believe that} in the study of symmetries, the Hamiltonian--horizontal decomposition {should be} preferred to the usual horizontal--vertical one~\eqref{eqn:splittingTM} {as} it behaves better with respect to the Lie brackets, as exemplified by Propositions~\ref{prop:contactomorphism} and~\ref{prop:automorphismhor}.

In practice, the relation between a vector field $\xi$ and the contact Hamiltonian system $(M,\eta,\cH)$ is often characterised by their Lie bracket $[\xi, X_{\cH}]$.
In the horizontal--vertical decomposition, we miss that the Hamiltonian vector fields are not authomorphisms of the vertical distribution, as the following computation shows:
\begin{equation}
	\left[f \Reeb, X_\cH\right] = - X_\cH (f) \Reeb + f X_\Reeb(\cH).
\end{equation}
So, after taking the lie brackets of $\xi$, the Reeb component will contribute terms to both the vertical and horizontal components of the bracket.
Instead, the Hamiltonian--horizontal decomposition reduces the computation to
\begin{equation}
	\label{eqn:splitcomm}
	[\xi,X_{\cH}] = [X_{\phi_{\xi}}, X_{\cH}] + [\delta_\xi, X_{\cH}].
\end{equation}
Here, the first term $[X_{\phi_{\xi}}, X_{\cH}]$ is a Hamiltonian vector field, by Propositions~\ref{prop:preservetionHam} and \ref{prop:hamiltonianpropbrackets}, and the second is horizontal, since by \eqref{eq:LieHameta}
\begin{equation}\label{eq:deltaxiham}
	\eta([\delta_\xi, X_{\cH}]) = \iota_{\delta_\xi} L_{X_{\cH}} \eta - L_{X_{\cH}} \iota_{\delta_\xi} \eta = - \Reeb(\cH) \iota_{\delta_\xi} \eta - L_{X_{\cH}} \iota_{\delta_\xi} \eta = 0.
\end{equation}
This means that the sum in equation~\eqref{eqn:splitcomm} is manifestly the Hamiltonian--horizontal decomposition of the commutator, i.e.,
\begin{equation}
	[\xi , X_\cH] = X_{-\{\phi_\xi, \cH\}_\eta} + \delta_{[\xi, X_\cH]}.
\end{equation}

One may still wonder what is the relation between the Hamiltonian--horizontal decomposition and the horizontal--vertical one from~\eqref{eqn:splittingTM}.

\begin{remark}\label{rmk:VertJacobi}
	Since the Reeb vector field is orthogonal to the distribution, any vertical vector field $Y$ on a contact manifold $(M,\eta)$ has to be parallel to it, and thus by~\eqref{eqn:contactHam-Lambda}, it can be written as
	\begin{equation}
		Y = f \Reeb = X_{-f} + \Lambda(df, \cdot ),
	\end{equation}
	where $\Lambda(\cdot, \cdot)$ is the skew bivector field of the Jacobi structure induced by $\eta$ as in Lemma~\ref{lemma:bivector}.

	Therefore, a vertical vector field has Hamiltonian--horizontal decomposition $Y = X_{\phi_Y} + \delta_Y$ with
	\begin{equation}
		X_{\phi_Y} = - X_{f}, \quad\mbox{and}\quad
		\delta_Y = + \Lambda(df, \cdot ).
	\end{equation}
	Here, we want to stress the key role of the Jacobi structure in extracting the Hamiltonian component from the vertical field.

	This centrality of the Jacobi structure will keep coming back also in the upcoming sections.
	By providing a natural way to identify functions in involution, symmetries, and dissipated quantities, it will play an important role in the study of the symmetries of contact Hamiltonian systems.
\end{remark}

\subsection{Hamiltonian--horizontal decomposition and tensor densities}\label{sec:hamhordensities}

The Hamiltonian--horizontal decomposition presented in the previous section strictly depends on the choice of the contact form. If we consider a contactomorphism $\eta \mapsto f \eta$ as in Proposition~\ref{prop:contactomorphism}, the Hamiltonian function will be rescaled by $f$ and a more complicated transformation law will hold for $\delta_Y$.
Nevertheless, it is possible to relax this dependence on the contact form and obtain an intrinsic representation of the splitting by introducing tensor densities, as done in~\cite[Chapter 7.5]{projective} and~\cite{vectorcontactstructure}.

We postpone the technical definitions to Appendix~\ref{app:tensor} and just recall here the relevant consequences in the case of exact contact manifolds.
What matters most for us is \red{equation~\eqref{eqn:app:tensorcontactexplaine} (see Appendix~\ref{app:tensor})}: on $(M,\eta)$ a tensor density $\phi$ of degree $\lambda$ is an object with the local expression
\begin{equation}
	\label{eqn:tendensity}
	\phi = f(x)\ \eta^{(n+1)\lambda}.
\end{equation}

\red{In particular, this means that the choice of a contact form $\eta$ allows to express tensor densities (as a module over the group of contactomorphisms) not in terms of the volume form induced by $\eta$ (see Remark~\ref{rmk:volume}) as one would usually do, but directly in terms of $\eta$ itself.
	This allows one to describe infinitesimal contactomorphisms as a special case of tensor densities.}

\begin{prop}[\cite{projective, vectorcontactstructure}]
	\label{prop:decohamden}
	The set of infinitesimal contactomorphisms is isomorphic to the set of tensor densities of degree $-\frac{1}{n+1}$.
\end{prop}

This means that we can use equation~\eqref{eqn:tendensity} to identify the Hamiltonian of a contact Hamiltonian system $(M,\eta,\cH)$ with a tensor density $\psi_{X_\cH}$ of degree $-\frac{1}{n+1}$, which is given explicitly by
\begin{equation}\label{eqn:hamtendensity}
	\psi_{X_\cH} = \eta(X_\cH)\ \eta^{-1} = -\cH \ \eta^{-1},
\end{equation}
where to make sense of $\eta^{-1}$ we view it as a tensor density of degree $-\frac{1}{n+1}$.
This representation is useful for two reasons~\cite{projective, vectorcontactstructure}:
\begin{itemize}
	\item first, the expression of the contact Hamiltonian on the right-hand side of \eqref{eqn:hamtendensity} is independent of the choice of the contact form;
	\item second, the Jacobi brackets of two contact Hamiltonian functions are the Poisson brackets of the corresponding densities.
\end{itemize}

This directly motivates the following definition.
\begin{defn}
	A tensor density $\sigma$ is an \emph{invariant tensor density}, if it is invariant under contactomorphism~\cite{vectorcontactstructure}.
\end{defn}

To find an invariant representation of horizontal vector fields of $(M, \eta)$ in terms of densities, it is convenient to introduce the set $\Omega_h^2(M)$ of 2-forms vanishing on the contact distribution, that is, 2-forms that can be expressed as $\eta \wedge \beta$ with $\beta \in \Omega^1(M)$.

\begin{prop}[\cite{projective, vectorcontactstructure}]
	\label{prop:decotangden}
	Let $(M, \eta)$ be an exact contact manifold.
	The set $\Delta = \ker \eta$ of horizontal vector fields is isomorphic to
	\begin{equation}
		\Omega_h^2(M) \otimes \mathcal{F}_{-\frac{2}{n+1}}(M),
	\end{equation}
	where $\mathcal{F}_{-\frac{2}{n+1}}$ denotes the set of tensor densities on $M$ of degree $-\frac{2}{n+1}$.
\end{prop}

With Propositions~\ref{prop:decohamden} and~\ref{prop:decotangden} at hand, we can revise the Hamiltonian--horizontal splitting and assign to each vector field two contact-invariant densities which represent, respectively, its horizontal and Hamiltonian parts.
\begin{prop}\label{prop:hamhordensity}
	Let $(M, \eta)$ be an exact contact manifold.
	The Hamiltonian--horizontal decomposition $\xi = X_{\phi_\xi} + \delta_\xi$ of any vector field $\xi\in\Gamma(TM)$ has an expression in terms of invariant tensor densities of the form $\psi_{\xi} + \sigma_\xi$, where
	\begin{itemize}
		\item the Hamiltonian component $\psi_\xi = \eta(\xi)\ \eta^{-1}$ is a tensor density of degree $-\frac{1}{n+1}$;
		\item and, the horizontal component $\sigma_\xi$ is the tensor density of degree $-\frac{2}{n+1}$ given by
		      \begin{equation}
			      \sigma_\xi = - \eta \wedge d\eta\left(Y+X_{\eta(Y)}, \cdot \right) \otimes \eta^{-\frac{2}{n+1}}.
		      \end{equation}
	\end{itemize}
\end{prop}
Thus, even though the \red{decomposition of the vector fields with respect to a fixed $\eta$ transforms in a complex way under contactomorphisms}, especially in its horizontal part, its \emph{density representatives} $(\psi_\xi,\sigma_\xi)$ are intrinsic and more convenient to work with.

\begin{example}
	\label{ex:transformationdensities}

	To give an intuition for the simplifications provided by the use of tensor densities, let us consider a very simple example: the standard $3$-dimensional contact manifold $(\mathbb{R}^3,\eta = ds - p dq)$.
	The Reeb vector field in this case is simply
	\begin{equation}
		\Reeb = \frac{\partial}{\partial s}.
	\end{equation}

	The contactomorphism
	\begin{equation}
		\phi_\alpha (q ,p, s) = (\xi, \pi, \tau) = \left(q, e^{\alpha s} p, \frac{1}{\alpha}e^{\alpha s}\right) ,
	\end{equation}
	transforms the standard 3-dimensional contact manifold to a new contact manifold $(\mathbb{R}^2 \times \mathbb{R}^+, \eta^\alpha = d\tau - \pi d\xi)$ such that
	\begin{align}
		(\phi_\alpha)^*(d \tau - \pi d \xi) & = e^{\alpha s} (ds - p dq),                               \\
		(\phi_\alpha)^*(d \xi \wedge d \pi) & = e^{\alpha s} (dq \wedge dp + \alpha\, p\ dq \wedge ds).
	\end{align}

	The Reeb vector field associated to $\eta^\alpha$ is:
	\begin{equation}
		\Reeb^\alpha = \frac{\partial}{\partial \tau}
	\end{equation}
	which in the original coordinates can be expressed as
	\begin{equation}
		\rho^\alpha := e^{-\alpha s} \frac{\partial }{\partial s} - \alpha p e^{-\alpha s} \frac{\partial}{\partial p}.
	\end{equation}
	Since $\rho^\alpha$ is also a contact Hamiltonian vector field on $(\mathbb{R}^3, \eta = ds - p dq)$ we can express it as a tensor density using equation~\eqref{eqn:hamtendensity}:
	\begin{equation}
		\psi_{\rho^\alpha} = \eta(\rho^\alpha) \eta^{-1} = e^{-\alpha s} \lvert dq \wedge dp \wedge ds \rvert^{-\frac12}.
	\end{equation}
	While the notation itself can look daunting if one is not yet used to it, the term in the absolute value is simply the volume form of the $-1/2$-density (fixed for the manifold) and its coefficient, and these are the only quantities that matter to represent our vector field. It is not just a coincidence that the coefficient is exactly the scaling factor of the contactomorphism, this is indeed a general property of tensor densities. It should also not be surprising that $\psi_{\rho^\alpha}$ is exactly the same for the Reeb vector field $\Reeb^\alpha$ on $(\mathbb{R}^2 \times \mathbb{R}^+, \eta^\alpha = d\tau - \pi d\xi)$, since we have already shown the invariance. Indeed, a direct computation shows:
	\begin{equation}
		\psi_{\rho^\alpha} = \lvert d\xi \wedge d\pi \wedge d\tau \rvert^{-\frac12} = \eta^\alpha(\Reeb^\alpha) (\eta^\alpha)^{-1} = \psi_{\Reeb^\alpha}.
	\end{equation}

	The same argument can be used for a horizontal vector field.
	For the sake of simplicity, consider the horizontal vector field
	\begin{equation}
		\frac{\partial}{\partial p} \mapsto \alpha \tau \frac{\partial}{\partial \pi}.
	\end{equation}
	Its tensorial component is given by
	\begin{equation}
		\sigma_{\frac{\partial}{\partial p}} = - \eta \wedge d\eta\left(\frac{\partial}{\partial p} \right) \otimes \eta^{-1} = ds \wedge dq \otimes \lvert dq \wedge dp \wedge ds \rvert^{-\frac12},
	\end{equation}
	that in the new coordinates becomes
	\begin{align}
		\sigma_{\alpha \tau \frac{\partial}{\partial \pi}} = - \eta^\alpha \wedge d\eta^\alpha\left(\alpha \tau \frac{\partial}{\partial \pi} \right) \otimes \left( \eta^\alpha\right)^{-1} & = - \alpha \tau d\tau \wedge d \xi \otimes \lvert d\xi \wedge d\pi \wedge d\tau \rvert^{-\frac12} \\&= ds \wedge dq \otimes \lvert dq \wedge dp \wedge ds \rvert^{-\frac12}.
	\end{align}
	Again, using the tensor densities the horizontal component $\sigma$ is invariant and thus there is no need to even perform any computation to know the pullback.

	While this is just a simple example, we hope it highlights where the use of tensor densities could become beneficial, especially in more complex problems which are however out of the scope of this paper.
\end{example}

The true advantage of tensor densities is less evident in examples where the system is fixed and the choice of coordinates is clear, such as those in Section~\ref{sec:applications}. However, their utility becomes significant when working with unconventional coordinate systems. For instance, we anticipate this approach will be valuable for rephrasing the integrable leaf of the Kepler problem on the Heisenberg group~\cite{Shanbrom_2014,Montgomery_2015} as an integrable contact Hamiltonian system. In that case, clarifying the properties of the system requires a complex coordinate transformation, rendering direct computations in the contact setting unwieldy, whereas the use of densities is expected to significantly simplify the analysis. Due to the complexity of this system, a detailed analysis is deferred to future work.

\red{Another case in which tensor densities simplify the analysis is in the presence of functions in involution with the Hamiltonian leading to non-exact contactomorphisms. In this case, the tensor density of the Hamiltonian remains invariant under the induced transformation, and thus remains constant along the flow. In this sense, tensor densities become interesting tools for providing intrinsic descriptions of conserved quantities in contact Hamiltonian systems. The study of dissipated quantities and contact integrability is still quite new, so we currently lack complex examples. In particular, we do not yet have an explicit example where this occurs, but there is no theoretical obstacle to its existence.}

\begin{remark}[Symplectization and tensor densities]\label{rem:sympl}
	The relation between the Hamiltonian--horizontal decomposition and tensor
	densities is also apparent from the perspective of symplectization. See Appendix~\ref{app:sympl} \red{for a detailed treatment, including relevant definitions and further details on the construction below}.

	\red{This connection can be made explicit by considering} an exact contact manifold $(M,\eta)$ and its symplectization
	$SM = M \times \mathbb{R}_r$ with Liouville form $\lambda = e^r \eta$ and
	symplectic form $\omega = d\lambda$. The Reeb lift has Hamiltonian
	$H_{SR} = \lambda(SR) = e^r$. For every $h \in C^\infty(M)$ we define the
	homogeneous lift $\widetilde h := h \, H_{SR}$. Then the Poisson bracket on
	$SM$ induces the Jacobi bracket on $M$ by
	\[
		\{h,k\}_\eta \;=\; \frac{\left\{\widetilde h,\widetilde k\right\}_{SM}}{H_{SR}}.
	\]

	In this way, contact Hamiltonians correspond to invariant tensor densities of
	degree $-\tfrac{1}{n+1}$, namely $\psi_{X_H} = -H \, \eta^{-1}$.
	Consequently, the Jacobi bracket of functions
	on $(M,\eta)$ coincides with the Poisson bracket of their homogeneous lifts on
	$SM$, or equivalently, with the Poisson bracket \red{of the associated tensor
		densities.}

	In this paper, we will not use this construction explicitly, as our focus is on developing the theory directly on the contact manifolds. However, it is important to note that this perspective can provide a useful point of view also in the study of contact systems.
\end{remark}

\section{Symmetries of contact Hamiltonian systems and their decomposition}
\label{sec:decsym}
In the symplectic setting, any function that is in involution with the Hamiltonian with respect to the Poisson brackets, is a constant of motion of the system and the Hamiltonian vector fields generated by functions that are in involution are pairwise commuting~\cite{de_Le_n_2020,Leon2019}.

For contact Hamiltonian systems, Proposition~\ref{prop:hamiltonianpropbrackets} tells us that the latter property remains true.
However, the quantities that give rise to symmetries are not necessarily conserved along the contact Hamiltonian flow: we already saw with equation~\eqref{eqn:nonconservation} that the contact Hamiltonian itself is in general not preserved along its own contact flow. This is why, symmetries of contact Hamiltonian systems have been introduced in terms of ``dissipated'' quantities~\cite{de_Le_n_2020,Leon2019}.

\begin{defn}
	Let $(M, \eta, \cH)$ be a contact Hamiltonian system. A function $f\in C^\infty(M)$ is a \emph{dissipated quantity} if
	\begin{equation}
		\label{eqn:comm0}
		\{f,\cH\}_\eta = 0.
	\end{equation}
	That is, if $f$ is in involution with the contact Hamiltonian $\cH$ with respect to the Jacobi brackets induced by the contact form $\eta$.
\end{defn}

A conserved quantity can nevertheless be recovered by considering ratios of dissipated quantities: if we consider equation~\eqref{eqn:comm0} and~\eqref{eqn:commnot0}, we can observe that any two quantities that are in involution, will dissipate with the same ratio.

\begin{prop}[\cite{Leon2019}]
	\label{prop:ratios}
	Let $(M, \eta, \cH)$ be a contact Hamiltonian system.
	If $f,g$ are two function in involution with the Hamiltonian $\cH$, their ratio
	\begin{equation}
		\kappa = \frac{f}{g}
	\end{equation}
	is conserved along the flow of $\cH$.
\end{prop}

In what follows, we will study other types of infinitesimal symmetries that are not necessarily Hamiltonian, and thus, not necessarily infinitesimal contact transformation: dynamical symmetries, Cartan symmetries (both introduced on the contact setting in~\cite{de_Le_n_2020}) and dynamical similarities~\cite{Bravetti2023,Sloan2018,Sloan2021}.

\begin{defn}
	\label{def:symmetries}
	Let $(M,\eta,\cH)$ a contact Hamiltonian system.
	A vector field $Y \in \Gamma(TM)$ is called
	\begin{itemize}
		\item \emph{dynamical symmetry} if
		      \begin{equation} \label{eqn:dynsym}
			      \eta\left( \left[Y,X_\cH\right] \right) = 0 \,;
		      \end{equation}
		\item \emph{dynamical similarity} if there exists $\lambda\in C^\infty(M)$, such that
		      \begin{equation}
			      \label{eqn:dynamicalsimilarity}
			      \left[Y,X_\cH\right] = \lambda X_\cH\,.
		      \end{equation}
		\item \emph{Cartan symmetry} if there exist $a,g \in C^\infty(M)$ such that
		      \begin{equation}
			      \label{eqn:cartandef}
			      \begin{cases}
				      L_Y \eta = a \eta + dg \\
				      Y(\cH) = a \cH + g R(\cH)
			      \end{cases}\,.
		      \end{equation}
	\end{itemize}
\end{defn}

The following proposition shows that, at least for the first two kinds of symmetries above, it is possible to recover ``dissipated'' quantities.
This is particularly interesting since the symmetry $Y$ is, in general, not a contact Hamiltonian vector field.

\begin{prop}[\cite{de_Le_n_2020}]
	\label{prop:deleonsymm}
	Let $(M,\eta,\cH)$ a contact Hamiltonian system. Let $Y\in\Gamma(TM)$, then the following hold.
	\begin{itemize}
		\item $Y$ is a dynamical symmetry if and only if $\{\eta(Y), \cH\}_\eta = 0$.
		\item $Y$ is a Cartan symmetry if and only if $\{\eta(Y) - g, \cH\}_\eta = 0$, where $g$ is defined in \eqref{eqn:cartandef}.
	\end{itemize}
	Moreover $Y$ is both a dynamical symmetry and a Cartan symmetry if and only if $d g=0$
	and $Y$ is a contact Hamiltonian vector field.
\end{prop}

This statement already shows that there are nontrivial relations between these various different notions of symmetry. In the following subsections we will try to characterize them more finely in terms of their decompositions.

\subsection{Dynamical symmetries}\label{sec:dynsym}

Dynamical symmetries are vector fields that commute with the contact Hamiltonian vector field modulo $\eta$, in the sense that their commutator with the contact Hamiltonian vector field belongs to the contact distribution.
The Hamiltonian--horizontal decomposition is particularly well suited to capture this.

\begin{thm}\label{thm:dynsym}
	Let $(M,\eta,\cH)$ be a contact Hamiltonian system, and $ Y \in \Gamma(TM)$. Then $Y$ is a dynamical symmetry for the system if and only if its Hamiltonian--horizontal decomposition $Y = X_{\phi_Y} + \delta_Y$ satisfies
	\begin{equation}
		\{\phi_Y, \cH\}_\eta = 0 \,.
	\end{equation}
\end{thm}

\begin{proof}
	If $Y$ is a dynamical symmetry, this follows immediately from the first statement of Proposition~\ref{prop:deleonsymm} after observing that $\phi_Y = -\eta(Y)$, see Proposition~\ref{prop:decompositionHami}.

	Conversely, let $Y = X_{\phi_Y} + \delta_Y$. Then
	\begin{equation}
		\eta([Y,X_{\cH}]) = \eta([X_{\phi_Y} + \delta_Y, X_{\cH}]) = \eta([X_{\phi_Y}, X_{\cH}]) = %
		\{\eta(Y), \cH\}_\eta = 0.
	\end{equation}
	This, by the first statement of Proposition~\ref{prop:deleonsymm}, concludes the proof.
\end{proof}

This means that $\phi_Y = - \eta(Y)$ is a dissipated quantity for the contact Hamiltonian system, and therefore by Proposition~\ref{prop:hamhordensity} we get the following corollary.

\begin{corollary}
	The invariant tensor density associated to a dynamical similarity is
	\begin{equation}
		\mathfrak{K}= - \phi_Y \ \eta^{-1}.
	\end{equation}
\end{corollary}

Theorem~\ref{thm:dynsym} makes it clear that the definition of dynamical symmetries restricts only the Hamiltonian part of the vector field, leaving the horizontal part free: roughly speaking, this reflects the fact that equation~\eqref{eqn:dynsym} provides only one condition out of the $2n+1$ required to specify a vector field on a contact manifold.

\subsection{Scaling symmetries}\label{sec:scalsymm}

The next symmetries that we want to investigate are dynamical similarities.
We first focus on a specific family, the so-called \emph{scaling symmetries}.

\begin{defn}\label{def:scaling}
	Let $(M,\eta,\cH)$ be a contact Hamiltonian system.
	A vector field $Y \in \Gamma(TM)$ is a scaling symmetry of $\cH$ if it satisfies the equation~\eqref{eqn:dynamicalsimilarity} with $\lambda$ constant.
\end{defn}
If $Y$ is a scaling symmetries, the distribution spanned by $X_\cH$ and $Y$ is involutive, so integrable by Frobenius theorem.

\begin{thm}
	\label{thm:scaling}
	If $Y \in \Gamma(TM)$ is a scaling symmetry of the contact Hamiltonian system $(M,\eta,\cH)$, then $Y$ is of the form
	\begin{equation}
		Y = X_{\phi_Y} + \delta_Y,
	\end{equation}
	with $\phi_Y = -\eta(Y)$ and $[\delta_Y,X_\cH] = 0$.
\end{thm}
\begin{proof}
	By Definition~\ref{def:scaling}, we have
	\begin{align}
		\label{eqn:proofscaling1}
		\left[Y, X_\cH \right]
		= \lambda X_\cH
		= X_{\lambda \cH}.
	\end{align}
	Since the right-hand side of equation~\eqref{eqn:proofscaling1} is Hamiltonian the left-hand side also has to be Hamiltonian. From Propositions~\ref{prop:automorphismhor}, \ref{prop:preservetionHam}, and~\ref{prop:decompositionHami}, we find that $\left[Y, X_\cH \right] = \left[X_{\phi_Y} , X_\cH \right] + \left[\delta_Y, X_\cH \right]$, in particular, this means that for Proposition~\ref{prop:automorphismhor}
	\begin{equation}
		0 = \delta_{\left[Y, X_\cH \right]} = \left[\delta_Y, X_\cH \right].
	\end{equation}
\end{proof}

The following corollary shows that, in fact, the horizontal component of a scaling symmetry can always be taken to be zero: that is, scaling symmetries are equivalent modulo their horizontal part.

\begin{corollary}
	\label{cor:scalinghamiltonian}
	Let $Y \in \Gamma(TM)$ be a scaling symmetry of a contact Hamiltonian system $(M,\eta,\cH)$. Then the following holds:
	\begin{itemize}
		\item $Y$ is purely Hamiltonian, i.e., such that $\delta_Y = 0$, if and only if $L_Y \eta = - \Reeb(\phi_Y) \eta$.
		\item The Hamiltonian vector field generated by the function $\phi_Y$ is always a scaling symmetry of the system.
	\end{itemize}
\end{corollary}

As a consequence of Theorem~\ref{thm:scaling} and Corollary~\ref{cor:scalinghamiltonian}, we can explicitly compute the action of the scaling symmetry $\tilde{Y}=X_{-\eta(Y)}$ on the Hamiltonian and the contact from:\begin{equation}
	\label{eqn:scalingHev}
	\begin{cases}
		\tilde{Y}(\cH) = \cH \left( \lambda - \Reeb(\phi_Y) \right) \\
		L_{\tilde{Y}} \eta = - \Reeb(\phi_Y) \eta
	\end{cases}\,.
\end{equation}
In~\cite{Bravetti2023} the definition of scaling symmetry $W$ on a symplectic Hamiltonian system $(N,\omega, \mathfrak{H})$, with $(N, \omega)$ a symplectic manifold (cfr. Appendix~\ref{app:sympl}), is given by
\begin{equation}
	\label{eqn:scalingSymp}
	\begin{cases}
		W(\mathfrak{H}) = \Lambda \mathfrak{H} \\
		L_W \omega = \omega
	\end{cases}.
\end{equation}
Then, $W$ is a scaling symmetry for $(N,\omega, \mathfrak{H})$ if and only if
\begin{equation}
	\left[W,X_{\mathfrak{H}} \right] = \left(\Lambda-1\right) X_{\mathfrak{H}},
\end{equation}
where $X_\mathfrak{H}$ represent the symplectic Hamiltonian vector field.
The proof follows immediately from the computation of the Lie derivative of $\omega$ with respect to $[W,X_{\mathfrak{H}}]$.

The main differences between~\eqref{eqn:scalingHev} and~\eqref{eqn:scalingSymp} are: the appearance, in the latter, of the Reeb vector field acting on the Hamiltonian, and the fact that the scaling symmetry in the symplectic setting is also a Liouville vector field~\cite{Bravetti2023,sympgeo} (see Definition~\ref{def:liouville-vf}).

These equations are even more relevant in the case our Hamiltonian function remains bounded away from zero. There, we can rescale the Hamiltonian to a Reeb system, as we have already seen in Section~\ref{subsec:chs}, and the scaling symmetries {simplify} further.
\begin{corollary}
	For a Reeb system $(M,\eta,-1)$, the scaling symmetries take the form
	\begin{equation}
		\Reeb(\phi_Y) = - \left\{\phi_Y,-1 \right\}_\eta = -\eta([Y, X_{-1}]) = \lambda,
	\end{equation}
	and therefore one can define the quantity
	\begin{equation}
		\kappa = \Reeb(\phi_Y),
	\end{equation}
	\red{whose value remains constant along the flow of the Hamiltonian}. In particular, this means that $-\kappa$ is a dissipated quantity.
\end{corollary}

We can use this result to find dissipated quantities for general contact Hamiltonian systems.

\begin{thm}
	\label{thm:scalingconstant}
	Let $(M,\eta,\cH)$ be a contact Hamiltonian system with a scaling symmetry  $Y$. Then
	\begin{equation}
		\Phi = - X_\cH\left(\frac{\phi_Y}{\cH}\right),
	\end{equation}
	\red{is constant along the flow} of $X_\cH$, that is,
	\begin{equation}
		X_\cH(\Phi) = 0,
	\end{equation} and thus $\Phi\cdot \cH$ is a dissipated quantity.
\end{thm}
\begin{proof}
	This follows from the previous corollary and the following explicit computation
	\begin{equation}
		X_\cH(\Phi) = X_\cH\left(-X_\cH\left(\frac{\phi_Y}{\cH}\right) \right)=X_\cH\left( \frac{- X_\cH(\phi_Y) \cH + \phi_Y X_\cH(\cH)}{\cH^2}\right) = - X_\cH \left(\lambda\right) = 0.
	\end{equation}
\end{proof}

\begin{example}
	Consider the contact Hamiltonian system $\left(\mathbb{R}^3, ds - p dq, \cH(q,p,s) = p\right)$.

	The function $\mathfrak{f}(q,p,s) = q p$ is the Hamiltonian function of a scaling symmetry with parameter $\lambda = 1$. Indeed,
	\begin{equation}
		\{p , q p\}_\eta = - p.
	\end{equation}

	By Theorem~\ref{thm:scalingconstant}, we find that the function $\Phi(q,p,s) = q$ is constant along the flow.
\end{example}

On compact manifolds, we can provide a nice characterisation of the parameter $\lambda$ of a scaling symmetry in terms of the flux of the ``weighted charge'' $\frac{\phi_Y}{\cH}$ of the symmetry across the manifold boundary and the ``internal correlation'' between $\frac{\phi_Y}{\cH}$ and the rate of change of the Hamiltonian along the Reeb $\Reeb(\cH)$.
While this is not necessarily a convenient formula to compute $\lambda$ in practice, it provides a both a nice interpretation of the scaling parameter in terms of geometric and dynamic quantities and, more importantly, a clear expression to argue for the existence of scaling symmetries with non-zero parameter $\lambda$.

\begin{thm}\label{thm:scalingconstantcompact}
	Let $(M,\eta,\cH)$ be a contact Hamiltonian system on a $2n + 1$ dimensional compact manifold $M$.
	If $Y$ is a scaling symmetry of the system with parameter $\lambda$, then
	\begin{equation}
		\lambda = -\frac{1}{\operatorname{vol}(M)}\left(
		(n+1) \int_M \frac{\phi_Y}{\cH}\; \Reeb(\cH) \, \Omega
		+\int_{\partial M} \frac{\phi_Y}{\cH}\; (\iota_{X_\cH} \Omega)
		\right),
	\end{equation}
	where $\Omega:=\eta\wedge (d\eta)^{\wedge n}$ the canonical volume on $M$ induced by the contact form and $\operatorname{vol}(M) := \int_M \Omega$.
\end{thm}
\begin{proof}
	Using~\eqref{eq:LieHameta} %
	and $L_{X_\cH}d\eta=d(L_{X_\cH}\eta)$ one obtains the well-known formula
	\[
		L_{X_\cH}\Omega
		= L_{X_\cH}\big(\eta\wedge (d\eta)^{\wedge n}\big)
		= -(n+1)\Reeb(\cH)\,\Omega,
	\]
	hence the divergence of $X_\cH$ with respect to $\Omega$ is
	\[
		\operatorname{div}_\Omega X_\cH = -(n+1)\Reeb(\cH).
	\]

	Let $\phi_Y = -\eta(Y)$ as in Theorem~\ref{thm:scaling} and observe that by Theorem~\ref{thm:scalingconstant}
	\begin{equation}\label{eqn:scalingint}
		\int_M X_\cH\left(\frac{\phi_Y}{\cH}\right) \, \Omega = \int_M (-\lambda) \, \Omega = -\lambda \, \text{Vol}(M).
	\end{equation}

	Observe that given any volume form $\sigma$, we have
	\begin{align}
		\int_M L_{X_\cH} \sigma = \underbrace{\int_M \iota_{H_X} d\sigma}_{=0 \text{ as }d\sigma = 0} + \int_M d(\iota_{X_\cH} \sigma)
		= \int_{\partial M} \iota_{X_\cH} \sigma. \label{eq:divthm}
	\end{align}

	Since $L_{X_\cH}\left(f \Omega\right) = X_\cH(f) \Omega + f L_{X_\cH} \Omega$, we can apply \eqref{eq:divthm} with $\sigma = \frac{\phi_Y}{\cH} \Omega$ and rewrite the left-hand side as
	\begin{align}
		\int_M X_\cH\left(\frac{\phi_Y}{\cH}\right) \, \Omega & = \int_M L_{X_\cH} \left(\frac{\phi_Y}{\cH} \Omega\right) - \int_M \frac{\phi_Y}{\cH} L_{X_\cH} \Omega                    \\
		                                                      & = \int_{\partial M} \frac{\phi_Y}{\cH}\; (\iota_{X_\cH} \Omega) + (n+1) \int_M \frac{\phi_Y}{\cH}\; \Reeb(\cH) \, \Omega.
	\end{align}

	Bringing everything together, we obtain
	\begin{align}
		-\lambda\,\operatorname{vol}(M)
		 & = \int_M X_\cH\left(\frac{\phi_Y}{\cH}\right)\,\Omega                                                                     \\
		 & = -\int_M \frac{\phi_Y}{\cH}\,L_{X_\cH}\Omega                                                                             \\
		 & = \int_{\partial M} \frac{\phi_Y}{\cH}\; (\iota_{X_\cH} \Omega) + (n+1) \int_M \frac{\phi_Y}{\cH}\; \Reeb(\cH) \, \Omega.
	\end{align}
	Dividing by $-\operatorname{vol}(M)$ concludes the proof.
\end{proof}

\begin{corollary}
	Under the same hypotheses of Theorem~\ref{thm:scalingconstantcompact},	if $M$ is a compact manifold without boundary, then
	\begin{equation}\label{eq:lambda_compact}
		\lambda = -\frac{n+1}{\operatorname{vol}(M)}\int_M \frac{\phi_Y}{\cH}\; \Reeb(\cH) \, \Omega.
	\end{equation}
\end{corollary}
\begin{proof}
	Note that \eqref{eq:divthm} is just the divergence theorem.
	If $M$ is compact without boundary, the right-hand side of \eqref{eq:divthm} vanishes simplifying the expression for $\lambda$ in Theorem~\ref{thm:scalingconstantcompact} as stated.
\end{proof}

Note that there is no a priori reason for $\lambda$ to vanish. For instance, if $\frac{\phi_Y}{\cH}\,\Reeb(\cH)$ has a definite sign, then $\lambda$ is non-zero. This is different from the compact symplectic case, where we necessarily have $\lambda = 0$ due to Liouville's theorem {(in agreement with \eqref{eq:lambda_compact})}, and the scaling symmetries with constant coefficients are exactly the conserved quantities of the system.

\subsection{Dynamical similarities}\label{sec:dynsimil}

When $\lambda$ is not a constant but a genuine function, the result of Theorem~\ref{thm:scaling} no longer holds.
Indeed, dynamical similarities in general have non-vanishing Hamiltonian and horizontal components.
To prove this we need a small preliminary Lemma.

\begin{lemma}
	\label{lemma:techsymm}
	Let $f,g$ be two smooth functions on an exact contact manifold $(M,\eta)$, then
	\begin{equation}
		X_{f g} = f X_g + g X_f + f g\, \Reeb. \label{eqn:lemmathesis_fg}
	\end{equation}
\end{lemma}
\begin{proof}
	Contracting $\eta$ by each side of equation~\eqref{eqn:lemmathesis_fg} we obtain
	\begin{equation}
		\begin{cases}
			\eta(X_{fg}) = - fg, \\
			\eta(f X_g + g X_f + f g\, \Reeb) = f\eta( X_g) + g \eta (X_f) + f g = - f g.
		\end{cases}
	\end{equation}
	While the corresponding contraction of $d\eta$ gives
	\begin{equation}
		\begin{cases}
			d\eta(X_{f g}, \cdot) = d(fg) - \Reeb(fg) \eta, \\
			d\eta(f X_g + g X_f + f g\, \Reeb, \cdot ) = f (dg - \Reeb(g) \eta) + g (df - \Reeb(f) \eta) = d(f g) - \Reeb(f g) \eta.
		\end{cases}
	\end{equation}
	This shows that both sides of the equation define the same contact Hamiltonian vector field, thus concluding the proof.
\end{proof}

Recalling Remark~\ref{rmk:VertJacobi}, the following result should not come completely as a surprise.

\begin{thm}
	\label{prop:dynsimcomm}
	Let $Y = X_{\phi_Y} + \delta_Y$ be the Hamiltonian--horizontal decomposition of a dynamical similarity of $(M,\eta,\cH)$.
	Then the scaling function $\lambda \in C^\infty(M)$ {satisfies}
	\begin{equation}
		\{\cH,\phi_Y \}_\eta = - \lambda \cH \qquad \text{and} \qquad
		[\delta_Y, X_{\cH}] = - \cH\, \Lambda(d\lambda, \cdot).
	\end{equation}
\end{thm}
\begin{proof}
	From equation~\eqref{eqn:dynamicalsimilarity}, we recover
	\begin{align}
		\eta \left(\left[Y,X_{\cH}\right]\right) & = - \lambda \cH,
	\end{align}
	which reads
	\begin{align}
		\eta([X_{\phi_Y}, X_{\cH}]) + \eta([\delta_Y,X_\cH]) & = - \lambda \cH
	\end{align}
	after expanding $Y = X_{\phi_Y} + \delta_Y$ in its Hamiltonian--horizontal decomposition.
	In the same way as~\eqref{eq:deltaxiham}, $\eta([\delta_Y,X_\cH]) = 0$, therefore
	\begin{equation}
		- \lambda \cH = \eta([X_{\phi_Y}, X_{\cH}]) = \iota_{X_{\phi_Y}} L_{X_{\cH}} \eta - L_{X_{\cH}} \iota_{X_{\phi_Y}} \eta.
	\end{equation}
	Equation~\eqref{eq:LieHameta}, the definition of contact Hamiltonian vector fields and equation~\eqref{eqn:commnot0} then imply
	\begin{align}
		- \lambda \cH = \Reeb(\cH) \phi_Y + X_{\cH}(\phi_Y) = \{\cH,\phi_Y\}_\eta.
	\end{align}

	Therefore, starting from the definition of scaling symmetry, one can use Proposition~\ref{prop:preservetionHam} and Lemma~\ref{lemma:techsymm} to show
	\begin{align}
		\lambda X_\cH = [Y,X_\cH]
		= X_{-\eta([Y,X_\cH])} + [\delta_Y,X_\cH]
		= X_{\lambda \cH} - \cH X_\lambda - \lambda \cH \Reeb.
	\end{align}
	Comparing terms in the last identity for the Hamiltonian--horizontal decomposition above, we get
	\begin{equation}
		[\delta_Y,X_{\cH}] = - \cH X_\lambda - \cH \lambda \Reeb = - \cH\, \Lambda(d \lambda, \cdot),
	\end{equation}
	which concludes the proof.
\end{proof}

The local straightening of a contact Hamiltonian system to a Reeb system can simplify computations by providing an explicit function for $\lambda$.
\begin{corollary}
	\label{prop:reebgsynsim}
	If $Y \in \Gamma(TM)$ is a dynamical similarity of the Reeb system $(M,\eta,-1)$, then $\lambda = -\Reeb(\phi_Y)$.
\end{corollary}
\begin{proof}
	From proposition~\ref{prop:dynsimcomm} follows that
	\begin{equation}
		\{\phi_Y,-1\}_\eta = \lambda,
	\end{equation}
	from which $\lambda=-\Reeb(\phi_Y)$.
\end{proof}

And, moreover, one can show that the dynamical similarity of a Reeb system is preserved by contactomorphisms.
\begin{prop}
	\label{prop:reebdynsim}
	Let $Y \in \Gamma(TM)$ be a dynamical similarity of $(M,\eta,-1)$.
	Let $\psi: M \to M$ be a contactomorphism such that $\psi^* \eta'= f_\psi \eta$.
	Then $\psi_*Y$ is a dynamical similarity for $(M,\eta',-f_{\psi} \circ \psi^{-1})$.
\end{prop}

\begin{proof}
	Lie brackets commute with pushforwards, thus
	\begin{equation}
		\psi_* [Y,\Reeb] =   [\psi_* Y, \psi_* \Reeb] = \lambda \circ \psi^{-1} \psi_* \Reeb.
	\end{equation}
	By Proposition~\ref{prop:contactomorphism} (and Remark~\ref{rmk:reebHam}), if $\psi^* \eta'= f_\psi \eta$, then
	\begin{equation}
		\psi_* \Reeb = - X'_{f_\psi \circ \psi^{-1}},
	\end{equation}
	where by $X'_f$ we mean the Hamiltonian vector field induced on $(M,\eta')$ by the function $f:M\to \mathbb{R}$. Then, if we denote $\tilde{Y}=\psi_* Y$, we get
	\begin{equation}
		[\tilde{Y}, X_{-f_\psi \circ \psi^{-1}}] = \lambda \circ\psi^{-1} X_{-f_\psi \circ \psi^{-1}}
	\end{equation}
	that is again a dynamical similarity of the system $(M,\eta',-f_\psi \circ \psi^{-1})$.
\end{proof}

Of course, this result can be extended locally to any contact Hamiltonian system.
\begin{prop}
	\label{prop:dynsympreservation}
	Let $Y\in\Gamma(TM)$ be a dynamical similarity of the contact Hamiltonian system $(M,\eta,\cH)$ such that $[Y,X_{\cH}] = \lambda X_{\cH}$ for some $\lambda \in C^{\infty}(M)$. If for any $x_0 \notin \cH^{-1}(0)$, there exist a neighborhood $U$ of $x_0$ and a contactomorphism
	\begin{equation}
		\psi: (U \subset M,\eta) \to (\psi(U),\eta'),
	\end{equation}
	such that $(\psi(U), \eta' = -\frac{1}{\cH} \eta, -1)$ is a Reeb system, then
	\begin{equation}
		\lambda = X_{\cH}\left(-\frac{\phi_Y}{\cH} \right).
	\end{equation}
\end{prop}
\begin{proof}
	Let us consider the Reeb system  $(\psi(U), \eta' = -\frac{1}{\cH} \eta, -1)$. By Proposition~\ref{prop:reebdynsim}, it has a dynamical symilarity $\psi_* Y$ that satisfies $[\psi_* Y, \Reeb'] = \lambda_{\Reeb'} \Reeb'$, where $\Reeb'$ is the Reeb associated to $\eta'$. By Corollary~\ref{prop:reebgsynsim}, $\lambda_{\Reeb'}= \Reeb'(\phi_{Y_{\Reeb}})$.
	Applying the pushforward of the inverse $\psi^{-1}_*$, we map
	\begin{equation}
		\lambda_{\Reeb} \circ \psi = \lambda_{\cH}.
	\end{equation}
	By combining the two
	\begin{equation}
		d( \phi_{Y_{\Reeb}} \circ \psi^{-1})(\psi^{-1}_* \Reeb ) = X_{\cH} \left(\phi_{Y_{\Reeb}} \right).
	\end{equation}
	Finally by Proposition~\ref{prop:contactomorphism}, we obtain
	\begin{equation}\label{eq:consqty}
		\phi_{Y_{\Reeb}} = -\frac{\phi_{Y_{\cH}}}{\cH},
	\end{equation}
	concluding the proof.
\end{proof}

\subsection{Cartan symmetries}\label{sec:cartan}

The definition of Cartan symmetries is more involved, and it is harder to give a direct geometric interpretation of the conditions in \eqref{eqn:cartandef}. Here the Hamiltonian--horizontal decomposition helps to clarify the role of the auxiliary function $g$.

\begin{thm}
	\label{thm:cartandeco}
	Let $(M,\eta,\cH)$ be a contact Hamiltonian system, and $Y \in \Gamma(TM)$. Then $Y$ is a Cartan symmetry if and only if its Hamiltonian--horizontal decomposition has the form
	\begin{equation}
		\label{eq:deccartsym}
		Y = X_{\phi_Y} + \Lambda(dg,\cdot),
	\end{equation}
	where $\Lambda$ is the skew-bivector field defining the natural Jacobi structure on $(M,\eta)$, and
	\begin{equation}
		\{\phi_Y + g,\cH\}_\eta = 0, \quad\mbox{and}\quad a=-\Reeb(\phi_Y + g).
	\end{equation}
\end{thm}

\begin{proof}
	Consider a Cartan symmetry $Y\in\Gamma(TM)$ of the contact Hamiltonian system $(M,\eta,\cH)$.
	It follows from Proposition~\ref{prop:deleonsymm} that $\eta(Y) - g$ is in involution with $\cH$, that is, $\{\phi_Y+g,\cH\}_\eta = 0$.

	It remains to prove that, in this case, $Y$ decomposes as in equation~\eqref{eq:deccartsym}.
	We know by Proposition~\ref{prop:decompositionHami} that, in general, $Y = X_{\phi_Y} + \delta_Y$. Then,
	\begin{equation}
		\label{eqn:lieder}
		a \eta + d g = L_Y \eta = L_{X_{\phi_Y}+\delta_Y} \eta = -\Reeb(\phi_Y) \eta + d\eta (\delta_Y,\cdot),
	\end{equation}
	where for the left-hand side we used the definition of Cartan symmetry and for the right-hand side \eqref{eq:LieHameta} and the fact that $\delta_Y$ is horizontal.
	Contracting by the Reeb vector field, we get that $a=-\Reeb(\phi_Y + g)$.

	Plugging the value of $a$ into equation~\eqref{eqn:lieder} then implies
	\begin{equation}
		d \eta (\delta_Y, \cdot) = dg - \Reeb(g)\eta = \iota_{X_g}d\eta = \iota_{\Lambda(dg, \cdot) - g \Reeb} d\eta = d\eta(\Lambda(dg, \cdot), \cdot).
	\end{equation}
	Since $X_{\phi_Y}$ is the Hamiltonian part of Y and $\Lambda(dg, \cdot)$ is horizontal, we get $\Lambda(dg, \cdot) = \delta_Y$ and therefore equation~\eqref{eq:deccartsym} holds.

	Conversely, let us consider $Y = X_{\phi_Y} + \Lambda(dg,\cdot)$. We want to show that $Y$ satisfies equations~\eqref{eqn:cartandef}.
	A computation analogous to~\eqref{eqn:lieder} yields
	\begin{align}
		L_Y \eta & = L_{X_{\phi_Y}} \eta + L_{\Lambda(dg,\cdot)} \eta                                                      \\
		         & = -\Reeb(\phi_Y) \eta + L_{\left(\Lambda(dg,\cdot)-g \Reeb\right)} \eta + L_{\left(g \Reeb\right)} \eta \\
		         & = -\Reeb(\phi_Y) \eta + L_{X_g} \eta + L_{g \Reeb} \eta                                                 \\
		         & = -\Reeb(\phi_Y) \eta -\Reeb(g) \eta + g L_\Reeb \eta + dg                                              \\
		         & = -\Reeb(\phi_Y) \eta - \Reeb(g) \eta + d g                                                             \\
		         & = -\Reeb(\phi_Y + g) \eta + d g\,,
	\end{align}
	thus showing that the first part of equations~\eqref{eqn:cartandef} holds with $a = -\Reeb(\phi_Y + g)$.
	It remains to show the second condition in Proposition~\ref{prop:deleonsymm}. An explicit computation gives
	\begin{align}
		Y (\cH) & = X_{\phi_Y}(\cH) + \Lambda(dg,d\cH)                                                    \\
		        & =\{\phi_Y, \cH\}_\eta - \cH\Reeb(\phi_Y) + \{g,\cH\}_\eta - \cH \Reeb(g) + g \Reeb(\cH) \\
		        & =\{\phi_Y + g, \cH\}_\eta - \cH\Reeb(\phi_Y + g) + g \Reeb(\cH)                         \\
		        & = a \cH + g \Reeb(\cH) + \{\phi_Y + g, \cH\}_\eta.
	\end{align}
	Since $\{\phi_Y + g, \cH\} = 0$ by hypothesis, we can conclude the proof.
\end{proof}

\begin{corollary}
	The invariant tensor densities associated with the Hamiltonian and horizontal components of a Cartan symmetry $Y$ are, respectively,
	\begin{equation}
		-\phi_Y \eta^{-1}
		\quad \mbox{and} \quad
		(\eta \wedge dg) \eta^{-\frac{2}{n+1}}.
	\end{equation}
	The invariant tensor density for the Cartan symmetry is
	\begin{equation}
		\mathfrak{K} = (\phi_Y + g)\ \eta^{-1}.
	\end{equation}
\end{corollary}
\begin{proof}
	Proposition~\ref{prop:hamhordensity} implies that the Hamiltonian part of $Y$ is equivalent to the tensor density $-\phi_Y \eta^{-1}$.
	By Lemma~\ref{lemma:bivector} and Theorem~\ref{thm:cartandeco}, the Horizontal part is equivalent to
	\begin{equation}
		\left(\eta \wedge d g \right) \eta^{-\frac{2}{n+1}}.
	\end{equation}

	Finally, using once more Theorem~\ref{thm:cartandeco}, we obtain the invariant tensor density.
\end{proof}

As we have observed in Theorem~\ref{thm:scalingconstant} and Proposition~\ref{prop:reebdynsim}, the quantity $X_\cH(\phi/\cH)$ seems to carry an important meaning. In fact, a similar characterization applies to Cartan symmetries. On a subset of $M$ in which $\cH \neq 0$, we can revisit the statement of Theorem~\ref{thm:cartandeco} as
\begin{equation}\label{eq:cartanhom}
	X_\cH\left(\frac{\phi_Y+g}{\cH}\right) = 0.
\end{equation}
Since $\delta_Y = \Lambda(dg, \cdot) = X_g + g \Reeb$, for any given $\phi$ the construction of a Cartan symmetry is then reduced to \eqref{eq:cartanhom} having a solution.

\section{An application to mechanical contact Hamiltonian systems}\label{sec:applications}

In this section, we will show how this kind of symmetry appears in mechanical contact Hamiltonian systems, and what this means in specific toy models.
Let $M = \mathbb{R}^3$ and $\eta = ds -p dq$.
Let us consider a mechanical contact Hamiltonian coupled with a monomial dissipation term:
\begin{equation}\label{eqn:Hamiltonian}
	\mathcal{H}=\frac{p^2}{2}+V(q)+\gamma s^\alpha.
\end{equation}
Assume that the potential $V(q)$ is homogeneous of degree $k$, i.e. $k\in\mathbb{N}$ is the smallest integer such that
\begin{equation}
	V(\mu q)=\mu^k V(q).
\end{equation}
By Euler's theorem for homogeneous functions,
\begin{equation}
	\frac{\partial V(q)}{\partial q_j}\,q^j=kV(q).
\end{equation}
\begin{prop}\label{prop:scalHamil}
	An infinitesimal transformation
	\begin{equation}\label{eqn:scaleqps}
		Y:=\xi_{q_i} q_i\frac{\partial}{\partial q_i}+\xi_{p_i} p_i\frac{\partial}{\partial p_i}+\xi_s s\frac{\partial}{\partial s},
	\end{equation}
	with real parameters $\xi_{q_i},\xi_{p_i},\xi_s$, is a contact transformation if and only if
	\begin{equation}\label{eqn:Yccondition}
		\xi_s-\xi_{p_i}-\xi_{q_i}=0 \quad i=1,\ldots,n.
	\end{equation}
\end{prop}
\begin{proof}
	Using Proposition~\ref{prop:decompositionHami} we compute
	\begin{equation}
		\eta(Y)=-h_Y=\xi_s-\xi_{q_i} p_i q_i,
	\end{equation}
	and recover the Hamiltonian part
	\begin{equation}
		X_{h_Y}=\xi_{q_i} q^i\frac{\partial}{\partial q_i}-(\xi_{q_i}-\xi_s)p_i\frac{\partial}{\partial p_i}+\xi_s s\frac{\partial}{\partial s}.
	\end{equation}
	The horizontal part vanishes if and only if $Y-X_{h_Y}=0$, that is, when \eqref{eqn:Yccondition} holds.
\end{proof}
Then, there are dissipations for which $Y$ defined in \eqref{eqn:scaleqps} is a scaling symmetry, see Definition~\ref{def:scaling}.
\begin{thm}
	\label{thm:homoscaling}
	Consider an Hamiltonian of the form \eqref{eqn:Hamiltonian},
	$\cH=\frac{p^2}{2}+V(q)+\gamma s^\alpha$, and
	let $Y$ be a vector field of the form \eqref{eqn:scaleqps}.
	If the degree of homogeneity of the potential $V(q)$ is $k$ and $\alpha=\dfrac{2k}{k+2}$, then $\cH$ admits the scaling symmetry $Y$ if
	\[
		\xi_{q_i}=\dfrac{2}{k+2}\,\xi_s \quad\mbox{and}\quad \xi_{p_i} = \dfrac{k}{k+2}\,\xi_s.
	\]
	In this case, the scaling constant is $\lambda=\dfrac{k-2}{k+2}\,\xi_s$.
\end{thm}
\begin{proof}
	Let $Y$ be as in Proposition~\ref{prop:scalHamil}.
	By Theorem~\ref{thm:scaling}, its horizontal part
	\begin{equation}
		\delta_Y=(\xi_s-\xi_{p_i}-\xi_{q_i})p_i\frac{\partial}{\partial p_i}
	\end{equation}
	should commute with $X_{\mathcal{H}}$. This happens only if
	\begin{equation}\label{eqn:condYscal}
		\xi_s-\xi_{p_i}-\xi_{q_i}=0.
	\end{equation}
	\begin{remark}
		Thus, in this particular case $Y$ is a scaling symmetry if and only if $Y$ is hamiltonian, see Proposition~\ref{prop:scalHamil}. This is not always the case, as we will see in the following sections.
	\end{remark}
	By restricting to such $Y$ and imposing the scaling condition
	\begin{equation}
		\{\mathcal{H},\eta(Y)\}=\lambda\mathcal{H},
	\end{equation}
	we obtain the algebraic relation
	\begin{equation}
		-\bigl((- \alpha +1)\xi_s+\lambda\bigr)\gamma s^{\alpha}
		+(\xi_{q_i} k-\lambda-\xi_s)V(q)
		-\frac{|p|^2(\lambda+2\xi_{q_i}-\xi_s)}{2}=0.
	\end{equation}
	This holds for $\alpha=\dfrac{2k}{k+2}$, $\xi_{q_i}=\dfrac{2}{k+2}\,\xi_s$, and $\lambda=\dfrac{k-2}{k+2}\,\xi_s$; $\xi_{p_i}$ is then fixed by \eqref{eqn:condYscal}.
\end{proof}
This proposition is useful to get an idea of how to look for scaling symmetries in mechanical contact Hamiltonian systems and already gives a hint on where to start. However, it does not provide a necessary condition, just a sufficient one, as we will see in Section~\ref{sec:freepart}.

\subsection{Damped harmonic oscillator}
\label{sec:dampedharmosc}
A classical example with Hamiltonian of the form \eqref{eqn:Hamiltonian} is the damped harmonic oscillator:
\begin{equation}
	\mathcal{H}=\frac{p^2}{2}+\frac{q^2}{2}+\gamma s.
\end{equation}
This is a contact Hamiltonian integrable system \cite{Zadra} whose dynamics is usually characterized by the parameter $\gamma$, where for $0<\gamma<2$ the dynamics is underdamped, for $\gamma>2$ it is overdamped, with a transition at the critical value $\gamma=2$.

It is shown in \cite{Zadra} that the damped harmonic oscillator admits the following function in involution:
\begin{equation}
	F:=\frac{1}{2}p^2+\frac{\gamma}{2}qp+\frac{1}{2}q^2.
\end{equation}

With Theorem~\ref{thm:homoscaling} at hand, we can immediately construct a function in involution with $\cH$. In this case, the potential $V(q) = q^2/2$ is homogeneous of degree $k=2$ while the dissipation exponent is $\alpha=1$, corresponding to a scaling symmetry (where we chose $\xi_s = 2$)
\[
	Y = q\frac{\partial}{\partial q}+ p\frac{\partial}{\partial p}+ 2 s\frac{\partial}{\partial s},
\]
with scaling constant $\lambda=0$.

This leads to the following alternative function in involution with $\cH$:
\begin{equation}
	G := \eta(Y) = 2s - q p.
\end{equation}
Note that the three cannot be independent and, indeed, $F$ and $G$ are related to $\cH$ by
\begin{equation}
	\cH = F + \frac{\gamma}{2} G.
\end{equation}

\subsection{Free particle with linear dissipation}\label{sec:freepart}
The simplest contact Hamiltonian example of \eqref{eqn:Hamiltonian} is the damped particle:
\begin{equation}
	\mathcal{H}(q,p,s):=\frac{p^2}{2m}+\gamma s,
	\qquad\Longrightarrow\qquad
	\begin{dcases}
		\dot q=\dfrac{p}{m}, \\
		\dot p=-\gamma p,    \\
		\dot s=\dfrac{p^2}{2m}-\gamma s.
	\end{dcases}
\end{equation}

\begin{lemma}[Free particle, linear dissipation]\label{lem:free-invariant}
	For $\cH$ on $(\mathbb R^3,\,\eta=ds-p\,dq)$ as defined above:
	\begin{enumerate}
		\item $p$ is in involution with $H$;
		\item $K_1:=\tfrac{m}{\gamma}q+p$ and $K_2:=\tfrac{\cH}{p}$ are constants of motion (on each component where $p$ has fixed sign).
	\end{enumerate}
\end{lemma}

\begin{proof}
	The first part follows directly from the cyclicity of $q$, the second follows by direct differentiation along $X_\cH$.
\end{proof}

We can use this to see that Theorem~\ref{thm:homoscaling} is sufficient but not necessary to find scaling symmetries: the free particle Hamiltonian does not satisfy the hypotheses of Theorem~\ref{thm:homoscaling}, but it is possible to recover a scaling symmetry for this Hamiltonian.

\begin{prop}
	$X_\Sigma$ with $\Sigma=\tfrac{\lambda}{\gamma}\,\ln|p|\, H$ is a scaling symmetry with parameter $\lambda$.
\end{prop}
\begin{proof}
	First observe that
	\begin{equation}
		\{\Sigma,\cH\}_\eta
		=\lambda\,\operatorname{sgn}(p)\!\left(\frac{p^2}{2m}+\gamma s\right)
		=\lambda\,\operatorname{sgn}(p)\,\cH.
	\end{equation}
	It then follows from the first point of Lemma~\ref{lem:free-invariant} and the fact that the sign of $p$ is constant, that the vector field $X_\Sigma$ is a scaling symmetry.
\end{proof}

Scaling symmetries, as characterized in Theorem~\ref{thm:scaling}, are composed of a Hamiltonian part that is a scaling symmetry and a commuting horizontal vector field. For instance, the horizontal vector field
\begin{equation}
	\label{eqn:horizontalsymm}
	\delta=\gamma\frac{\partial}{\partial p}+\Bigl(\frac{\partial}{\partial q}+p\frac{\partial}{\partial s}\Bigr)
\end{equation}
commutes with $X_{\mathcal{H}}$. Hence, every vector field of the form
\begin{equation}
	X_\Sigma+\delta
\end{equation}
is a scaling symmetry of $\cH$. We can push this further by using the two constants of motion $K_1$ and $K_2$ to construct a family of horizontal vector fields commuting with $X_{\cH}$.

\begin{prop}
	A horizontal vector field of the form
	\begin{align}
		Z={} & -p\,f_{1}\!\Bigl(\gamma q+p,\frac{qp-2s}{2p}\Bigr)\,\frac{\partial}{\partial p}\notag \\[4pt]
		     & \qquad
		+\Bigl(-f_{1}\!\Bigl(\gamma q+p,\frac{qp-2s}{2p}\Bigr)\,x
		+f_{2}\!\Bigl(\gamma q+p,\frac{qp-2s}{2p}\Bigr)\Bigr)
		\Bigl(p\frac{\partial}{\partial s}+\frac{\partial}{\partial q}\Bigr),
	\end{align}
	where $f_1(x,y)$ and $f_2(x,y)$ are smooth functions, commutes with $X_{\mathcal{H}}$.
	Consequently $X_\Sigma+Z$ is a scaling symmetry for $X_{\mathcal{H}}$.
\end{prop}
\begin{proof}
	The result follows from a direct computation of the Lie bracket $[V,X_{\mathcal{H}}]$, noticing that the arguments of $f_1$ and $f_2$ are constants of motion for $X_{\mathcal{H}}$.
\end{proof}

\begin{example}
	To illustrate a practical advantage of the tensor density description, we revisit the free particle with linear dissipation $(\mathbb{R}^3,\eta = ds - p dq, \cH = \frac{p^2}{2m} + \gamma s)$. Performing the contact transformation $\phi_\alpha$ defined in Example~\ref{ex:transformationdensities}, namely,
	\begin{equation}
		\phi_\alpha (q ,p, s) = (\xi, \pi, \tau) = (q, e^{\alpha s} p, \frac{1}{\alpha}e^{\alpha s}),
	\end{equation}
	we obtain a new contact Hamiltonian system
	\begin{equation}
		\left(\mathbb{R}^2 \times \mathbb{R}^+, \eta^\alpha = d\tau - \pi d\xi,\ \cH^\alpha = \frac{\pi^2}{2 m \tau \alpha} + \gamma \tau \ln(\alpha \tau) \right).
	\end{equation}

	Crucially, the Hamiltonian function transforms non-trivially involving the conformal factor $\Omega = \alpha \tau$:
	\begin{equation}
		\cH(q,p,s) = \frac{\cH^\alpha}{\tau \alpha} \circ \phi_\alpha (q,p,s).
	\end{equation}
	In contrast, the tensor density formulation provides an invariant description. The density associated with the Hamiltonian vector field takes the same form in both coordinate systems:
	\begin{equation}
		\psi_{X_\cH} = -\left(\frac{p^2}{2m} + \gamma s \right) \lvert dq \wedge dp \wedge ds \rvert^{-\frac12} = - \left( \frac{\pi^2}{2 m \tau \alpha} + \gamma \tau \ln(\alpha \tau) \right) \lvert d\xi \wedge d \pi \wedge d\tau \rvert^{-\frac12},
	\end{equation}
	where in both cases the coefficient of the density is simply the Hamiltonian function expressed in the corresponding coordinates.

	Similarly, the horizontal vector field $\delta$ from equation~\eqref{eqn:horizontalsymm} commutes with the Hamiltonian. Its representation as a tensor density $\sigma_\delta$ makes its invariance under the transformation manifest:
	\begin{align}
		\sigma_{\delta} & =  (- \gamma \alpha \tau\, d \tau \wedge d \xi + d\tau \wedge d \pi - \pi\, d\xi \wedge d \pi) \otimes \lvert d\xi \wedge d\pi \wedge d\tau \rvert^{-\frac12} \\
		                & = \left( - \gamma\, ds \wedge dq +ds \wedge dp - p\, dq \wedge dp \right) \otimes \lvert dq \wedge dp \wedge ds \rvert^{-\frac12}.
	\end{align}
	One still has to check that $\sigma_\delta$ is invariant under the Hamiltonian flow, that is, that $\mathcal{L}_{X_\cH} \sigma_\delta = 0$.

	While computations in this specific example are easily manageable in coordinates, the tensor density formalism offers a distinct advantage: one can verify this conservation law directly in any chart without having to track how the conformal factors modify the brackets.
	This separation of intrinsic geometric properties from coordinate-dependent artifacts becomes particularly relevant in more complex systems, where explicit coordinate transformations are less tractable.
\end{example}

\subsection{Hamiltonian with quadratic ''action``}
Another common example is the contact Hamiltonian with quadratic action, initially explored in the context of cosmology \cite{Bravetti_2021}:
\begin{equation}
	\mathcal{H}(q,p,s):=\frac{p^2}{2m}+\mu\frac{q^2}{2}+\frac{\gamma}{2}s^2-R.
\end{equation}
The equations of motion read
\begin{equation}
	\begin{dcases}
		\dot q=\dfrac{p}{m},          \\
		\dot p=-\mu q - \gamma p\, s, \\
		\dot s=\dfrac{p^2}{2m}-\mu\dfrac{q^2}{2}-\dfrac{\gamma}{2}s^2+R.
	\end{dcases}
\end{equation}

When $\mu = 0$ it is not difficult to find a function in involution with $\cH$.
For one, as in the previous case, $q$ is a cyclic variable and therefore $p$ is a dissipated quantity.

\medskip
In the specific case in which $\mu = R = 0$, the Hamiltonian vector field $X_s = -p \partial_p -s\partial_s$ generated by $(q,p,s) \mapsto s$ is a scaling symmetry with $\lambda=-1$, since
\begin{equation}
	\{\mathcal{H},s\}_\eta=\mathcal{H}.
\end{equation}
Moreover, the horizontal vector field
\begin{equation}
	Y=sp\frac{\partial}{\partial p}
	+\Bigl(\frac{-p^2+s^2}{p}\Bigr)\Bigl(p\frac{\partial}{\partial s}+\frac{\partial}{\partial q}\Bigr)
\end{equation}
commutes with $X_{\mathcal{H}}$. Consequently, for any $\kappa\in\mathbb{R}$ the vector field
\begin{equation}
	X_s+\kappa\Bigl(sp\frac{\partial}{\partial p}
	+\Bigl(\frac{-p^2+s^2}{p}\Bigr)\Bigl(p\frac{\partial}{\partial s}+\frac{\partial}{\partial q}\Bigr)\Bigr)
\end{equation}
is a scaling symmetry of $\mathcal{H}(q,p,s)=\tfrac12 p^2+\tfrac{\gamma}{2}s^2$.

\medskip
This example is particularly interesting since it shows that, even if one can find a function in involution with $\cH$, and consequently the system presents an integrable dynamics (in a sense that will be made precise in Definition~\ref{def:ccis}), the Hamiltonian vector field needs not be complete, as observed in~\cite{Liu2021OrbitalDO}. In fact, if $\mu = 0$, all trajectories that start from a point with coordinates $(q,0,s)$ will blow up in finite time. This changes drastically if $\mu \neq 0$ and $\gamma, R > 0$, in which case the dynamics remains ``stiff'' but all trajectories are complete, see~\cite{Bravetti_2021} for more details.

\section{Integrability of contact Hamiltonian systems and their decomposition}\label{sec:integra}
In the classical Hamiltonian setting, Noether's theorem gives a profound link between symmetries and conserved quantity, ultimately linking integrability to the presence of a sufficient number of independent constants of motion.
In the contact case this is less clear, and in the past year attempts have already been made to give a notion of contact integrability and relate it with symmetries or dissipated quantities, see e.g.~\cite{de_Le_n_2020, Jovanovi__2015, Grillo_2020} and the references therein.
In this section, we want to revisit this discussion in view of the decomposition presented above.

\begin{defn}
	\label{def:ccis}
	A \emph{complete contact integrable system} is a triple $\left(M,\eta,\{f_0, \ldots,f_n\} \right)$, where $M$ is a $2n+1$-dimensional contact manifold endowed by a contact form $\eta$, and $\{f_0, \ldots, f_n\}$ is a set of functions such that
	\begin{itemize}
		\item $\{f_i,f_j\}_\eta = 0$ for all $i,j \in \{0,\ldots,n\}$
		\item and $X_{f_0}, \ldots, X_{f_n}$ are linear independent almost everywhere.
	\end{itemize}
\end{defn}
The symmetry results of Section~\ref{sec:decsym} naturally yield functions in involution: dynamical symmetries provide dissipated quantities (Proposition~\ref{prop:deleonsymm}), and scaling/dynamical similarities can generate new ones from existing conserved quantities (Theorem~\ref{thm:scalingconstant} and Proposition~\ref{prop:const1dynsim}). The question that remains to explore is how such families can lead to complete contact integrability.

\medskip
In the symplectic version of Definition~\ref{def:ccis}, the requirement of independence of the functions $\{h_i\}_{i=1,\ldots,n}$ is written as
\begin{equation}
	\bigwedge_{i=1}^{n} d h_i \neq 0 \quad\mbox{almost everywhere~\cite{Knauf2018}.}
\end{equation}
This can be reformulated by considering the natural volume form induced by the symplectic form $\omega$ as
\begin{equation}
	\iota_{V_{h_1}, \ldots, V_{h_n}} \omega^{\wedge n} := \omega^{\wedge n}(V_{h_1}, \ldots, V_{h_n}, \cdot, \cdots, \cdot) \neq 0,
\end{equation}
where we denoted $V_h$ the (symplectic) Hamiltonian vector field associated to $h$, to distinguish it from their contact counterparts. See Appendix~\ref{app:sympl} for a small review of symplectic mechanics.

\begin{example}
	On a $4$-dimensional symplectic manifold $(P, \omega)$, a complete integrable system is defined by two functions in involution:
	\begin{align}
		\iota_{V_{f_1},V_{f_2}} \omega^{\wedge 2}
		 & = 2\ \omega(V_{f_1},V_{f_2}) \omega + 2\ \iota_{V_{f_1}} \omega \wedge \iota_{V_{f_2}} \omega \\
		 & = 2\{f_1,f_2\}_{PB} \omega + 2\ df_1 \wedge df_2                                              \\
		 & = 2\ df_1 \wedge df_2,
	\end{align}
	where by $\{\cdot, \cdot\}_{PB}$ we refer to the Poisson bracket induced by the symplectic form $\omega$.
\end{example}

The construction presented above extends to the contact case, were we can express the linear independence condition by via the natural volume form~\eqref{eqns:contactvol} induced by the contact form $\eta$, i.e.,
\begin{equation}
	\label{eqn:indeqn}
	\iota_{X_{f_0},X_{f_1}, \ldots, X_{f_n}} (\eta \wedge (d\eta)^{\wedge n}) \neq 0.
\end{equation}
The kernel of the $n$-form~\eqref{eqn:indeqn} corresponds to the distribution generated by the Hamiltonian vector field identified by the integrable system.

Moreover, the linear independence of a set of non-commuting vector fields $\{Y_j\}_{j\leq J}$ can be checked by the analogous condition
\begin{equation}
	\iota_{Y_0, Y_1, \ldots, Y_{J}}(\eta \wedge (d\eta)^{\wedge n}) \neq 0.
\end{equation}

\begin{lemma}
	\label{lemma:li}
	Let $(M,\eta)$ be a $(2n+1)$-dimensional exact contact manifold, and assume there is $m\leq n$ such that $\{f_0, \ldots, f_m \}$ is a set of functions in involution whose contact Hamiltonian vector fields are linearly independent. Then
	\begin{align}
		\label{eqn:thesisind}
		 & \iota_{X_{f_0},X_{f_1}, \ldots, X_{f_n}}\left(\eta \wedge (d\eta)^{\wedge n} \right) =                                                                                          \\
		 & \quad =
		\frac{n!}{(n-m)!} \left(-1 \right)^{\frac12(m+1)(m+2)} \bigg[\sum_{j=0}^{m} \left(-1 \right)^{j m} f_j \bigwedge_{r=1}^{m} df_{j+r \pmod{m+1}} \bigg] \wedge d \eta^{\wedge (n-m)} \\
		 & \qquad + \frac{n!}{(n-m-1)!} \left(-1 \right)^{\frac12 m(m+1)} \bigwedge_{j=0}^{m} df_j \wedge \eta \wedge d\eta^{\wedge (n-m-1)}.
	\end{align}
\end{lemma}

Before entering into the details of the proof, it is worth exploring the consequences of this technical lemma.

\begin{thm}\label{thm:independence}
	Let $(M,\eta,\{f_0,f_1,\cdots,f_n\})$ be a complete integrable contact Hamiltonian system on a $(2n+1)$-dimensional contact manifold. Then
	\begin{align}
		\iota_{X_{f_0},X_{f_1}, \ldots, X_{f_n}} \left(\eta \wedge (d\eta)^{\wedge n} \right) = n! (-1)^{\frac12 (n+1)(n+2)} \bigg[\sum_{j=0}^{n} \left(-1 \right)^{j n} f_j \bigwedge_{r=1}^{n} df_{j+r \pmod{n+1}} \bigg]
	\end{align}
\end{thm}

\begin{proof}
	Fix $m = n$, apply Lemma~\ref{lemma:li} and observe that the second expression in equation~\eqref{eqn:thesisind} vanishes.
\end{proof}

In a three-dimensional contact manifold, a complete integrable contact Hamiltonian system is specified by $(M,\eta, \{f_0, f_1\})$. It has a constant of motion, when it can be defined, $\kappa = f_0 / f_1$. The $1$-form defined by equation~\eqref{eqn:indeqn} in this case reads
\begin{equation}
	\label{eqn:independence}
	\iota_{X_{f_0},X_{f_1}}\left(\eta \wedge d\eta\right) = - f_0 d f_{1} + f_1 d f_{0} + \{f_0, f_1\}_\eta \ \eta = - f_0 d f_{1} + f_1 d f_{0}.
\end{equation}
This can be rewritten as
\begin{equation}
	\iota_{X_{f_0},X_{f_1}}\left(\eta \wedge d\eta\right) = - \frac{1}{f_0^2}\,d\kappa.
\end{equation}
Note that we use the involutivity only in the last identity in~\eqref{eqn:independence}. The first identity holds for any pairs of Hamiltonian vector fields.

Integrability requires $f_0$ and $f_1$ to be linearly independent. However, note that in the submanifold $f_0 = f_1 = 0$, if it exists, the two Hamiltonian vector fields $X_{f_0}$ and $X_{f_1}$ are always linearly dependent.
This means that \emph{a complete integrable contact Hamiltonian system cannot admit $0$-levelsets}. However, outside the $0$-levelset, a system can be locally completely integrable. In this case, $0$-levelsets identify equilibrium regions that could still be reached asymptotically. This can be described more generally in the sense of Sussmann integrability~\cite{Sussmann_1973} but we will not pursue this direction here.

\begin{example}
	To see the result in practice, let's consider the contact manifold $(M, \eta)$ given by
	\begin{equation}
		M = \mathbb{R}^{5}, \qquad \eta = ds - p_1 dq_1 - p_2 dq_2,
	\end{equation}
	and the Hamiltonian function
	\begin{equation}
		\cH := q_1 p_1 + q_2 p_2 + \gamma s,
	\end{equation}
	where $0\neq \gamma \in \mathbb{R}$ is a fixed parameter.
	Along with the functions
	\begin{equation}
		f_k := q_k p_k \qquad k = 1,2,
	\end{equation}
	we obtain a contact complete integrable system.

	First, we verify involutivity by computing the Jacobi brackets.
	For $k=1,2$, we have
	\begin{equation}
		\left\{\cH , f_k \right\} = - X_{f_k}(\cH) = - p_k \frac{\partial \cH}{\partial p_k} + q_k \frac{\partial \cH}{\partial q_k} = - p_k q_k + q_k p_k = 0,
	\end{equation}
	and obviously $\left\{f_1,f_2\right\} = 0$ since they depend on disjoint sets of variables.

	Next, using Theorem~\ref{thm:independence}, we can show linear independence almost everywhere by checking that
	\begin{equation}
		\Upsilon := \sum_{j=0}^{2} f_j \, df_{j+1 \pmod{3}} \wedge df_{j+2 \pmod{3}} \neq 0,
	\end{equation}
	almost everywhere, where we set $f_0 = \cH$.

	Computing the differentials
	\begin{equation}
		\begin{dcases}
			df_0 = p_1 dq_1 + p_2 dq_2 + q_1 dp_1 + q_2 dp_2 +\gamma ds, \\
			df_1 = p_1 dq_1 + q_1 dp_1,                                  \\
			df_2 = p_2 dq_2 + q_2 dp_2,
		\end{dcases}
	\end{equation}
	and substituting them into the expression for $\Upsilon$, we obtain a non-vanishing 2-form.
	In coordinates $(q_1, q_2, p_1, p_2, s)$, this 2-form can be represented by a skew-symmetric matrix, namely
	\begin{equation}
		\bigg[\sum_{j=0}^{2} f_j \bigwedge_{r=1}^{2} df_{j+r \pmod{3}} \bigg] = \begin{pmatrix}
			0                  & p_1 p_2 \gamma s   & p_1 q_2 \gamma s     & -q_2 p_2 \gamma p_1  & 0                   \\
			- p_1 p_2 \gamma s & 0                  & q_1 p_2 \gamma s     & q_1 q_2 \gamma s     & -q_2 p_2 \gamma q_1 \\
			- p_1 q_2 \gamma s & - q_1 p_2 \gamma s & 0                    & 0                    & q_1 p_1 p_2 \gamma  \\
			q_2 p_2 \gamma p_1 & - q_1 q_2 \gamma s & 0                    & 0                    & q_1 p_1 q_2 \gamma  \\
			0                  & q_2 p_2 \gamma q_1 & - q_1 p_1 p_2 \gamma & - q_1 p_1 q_2 \gamma & 0
		\end{pmatrix}.
	\end{equation}
	A direct check shows that this form vanishes on a subset of measure zero (specifically at the origin, since $\gamma\neq0$), thus satisfying the requirements of the theorem.
\end{example}

\begin{thm}\label{thm:scalingindependence}
	Let $Y$ be a scaling symmetry of the Hamiltonian system $(M,\eta,\cH)$ and assume that there exists a contact Hamiltonian function $f$ in involution with $\cH$, i.e., $\{f,\cH\}_\eta = 0$.
	Then
	\begin{equation}
		\label{eqn:definitionpsiscaling}
		\psi=\left\{\phi_Y, f \right\}_\eta
	\end{equation}
	is a function in involution with $\cH$.

	Moreover, the scaling symmetry $Y$ is linearly independent from $X_\cH$ and $X_f$ if
	\begin{align}
		 & \cH \left(\psi + f \lambda \right) d \eta\wedge d\eta^{\wedge (n-1)}                               \\
		 & \quad + (n-1) \bigg[
			(\cH df -f d\cH) \wedge d\phi_Y +  (\psi - \lambda \cH) df \wedge \eta + \phi_Y d \cH \wedge df
		\bigg] \wedge d\eta^{\wedge(n-2)}                                                                     \\
		 & \quad + (n-1) (n-2)\; d \phi_Y\wedge d\cH \wedge df \wedge \eta \wedge d\eta^{\wedge (n-3)} \neq 0
	\end{align}
\end{thm}

\begin{proof}
	The first identity, \eqref{eqn:definitionpsiscaling}, can be proved via the Jacobi identity,
	\begin{align}
		[Y,[X_f,X_\cH]] + [X_\cH,[X_f,Y]] + [X_f,[Y,X_\cH]] = 0.
	\end{align}
	Applying the definition of scaling symmetry, this reduces to
	\begin{equation}
		[X_\cH,[X_f,Y]]                                       = 0.
	\end{equation}
	Then, it follows from Corollary~\ref{cor:scalinghamiltonian} that $X_{\phi_Y}$ is itself a scaling symmetry and $\psi = \{\phi_Y,f\}_\eta$ is in involution with $\cH$.

	Theorem~\ref{thm:scaling} shows that $\psi = \{\phi_Y,f\}_\eta$ is in involution with $\cH$. Indeed,
	\begin{align}
		0 & = \eta( [X_\cH,[X_f,Y]] )                                      \\
		  & = L_{X_\cH} (\eta([X_f,Y])) - \iota_{[X_f,Y]} L_{X_\cH} \eta = \\
		  & = L_{X_\cH} \{f,\phi_Y\} + \Reeb(\cH) \iota_{[X_f,Y]} \eta =   \\
		  & = \{\cH, \{f,\phi_Y\}\}.
	\end{align}

	We are left to prove the second part of the statement.
	Lemma~\ref{lemma:li} where $m=2$ and for the set of functions $\{\cH, f\}$ implies
	\begin{equation}
		\label{eqn:liscaling}
		\iota_{X_\cH, X_f}\left(\eta \wedge d\eta^{\wedge n}\right) =
		- n \left[
			\left(\cH df - f d\cH\right) \wedge d\eta
			+ (n-1) d\cH \wedge df \wedge \eta
			\right] \wedge d \eta^{\wedge (n-2)}.
	\end{equation}

	Contracting equation~\eqref{eqn:liscaling} with $X_{\phi_Y}$, and using equations~\eqref{eqn:definitionpsiscaling} and~\eqref{eqn:scalingHev} we obtain:
	\begin{align}
		\iota_{X_\cH, X_f}\left(\eta \wedge d\eta^{\wedge n}\right)
		 & = n \cH \left(\psi + f \lambda \right) d \eta^{\wedge (n-1)}                                                 \\
		 & \quad + (n-1) \bigg[ (\cH df -f d\cH) \wedge d\phi_Y                                                         \\
		 & \quad\quad + (\psi - \lambda \cH) df \wedge \eta + \phi_Y d \cH \wedge df \bigg] \wedge d\eta^{\wedge (n-2)} \\
		 & \quad + (n-1) (n-2) d \phi_Y\wedge d\cH \wedge df \wedge \eta \wedge d\eta^{\wedge (n-3)};
	\end{align}
	since $n \neq 0$ (otherwise the contact manifold would be $1$-dimensional) and $d\eta\neq 0$, this concludes the proof.
\end{proof}

A similar result holds in the more general case of dynamical similarities.
In this case, however, it does not directly hold for functions in involution but for the constants of motion.

\begin{prop}
	\label{prop:const1dynsim}
	Let $Y\in\Gamma(TM)$ be a dynamical similarity of the contact Hamiltonian system $(M,\eta,\cH)$. If there exists a function $f \in C^{\infty}(M)$ that is a constant of motion of the Hamiltonian, i.e., $X_{\cH} f = 0$, then $\psi = Y\left({f}\right)$ is also a constant of motion for $\cH$.
\end{prop}

\begin{proof}
	By definition $[Y,X_{\cH}] (f) = g X_{\cH}(f) =0$, that is, $X_{\cH} \left( Y(f) \right) = 0$.
\end{proof}

\medskip
We conclude this section with the proof of Lemma~\ref{lemma:li}.
\begin{proof}[Proof of Lemma~\ref{lemma:li}]
	We will prove the statement by induction. Let us consider $m=0$, so the set of functions consists only of $f_0$. The contraction of $X_{f_0}$ on the volume form is:
	\begin{align}
		\iota_{X_{f_0}} (\eta \wedge d\eta^{\wedge n})
		 & =  - f_0 d\eta^{\wedge n} +n\ \left( df_0 - \Reeb(f_0)\eta \right)\wedge \eta \wedge d\eta^{\wedge (n-1)} \\
		 & = - f_0 d\eta^{\wedge n} + n\ df_0 \wedge \eta \wedge d\eta^{\wedge (n-1)}                                \\
		 & = \left(-f_0 d\eta + n\ df_0 \wedge \eta\right) d\eta^{\wedge (n-1)}.
	\end{align}

	Consider now $m=k-1$, the term would be
	\begin{align}
		\label{eqn:stepk-1}
		 & \iota_{X_{f_0},X_{f_1}, \ldots, X_{f_{k-1}}}\left(\eta \wedge (d\eta)^{\wedge n} \right) =                                                                                                                        \\
		 & \quad = \left[\binom{n}{k-1} (k-1)! \right] \left(-1 \right)^{\frac12 k(k+1)} \bigg[\sum_{j=0}^{k-1} \left(-1 \right)^{j (k-1)} f_j \bigwedge_{r=1}^{k-1} df_{j+r \pmod{k}} \bigg] \wedge d \eta^{\wedge (n-k+1)} \\
		 & \qquad + \left[\binom{n}{k} k! \right] \left(-1 \right)^{\frac12 k\left(k-1\right)} \bigwedge_{j=0}^{k-1} df_j \wedge \eta \wedge d\eta^{\wedge (n-k)}.
	\end{align}

	First observe that
	\begin{align}
		 & \bigg[\sum_{j=0}^{k-1} \left(-1 \right)^{j (k-1)} f_j \bigwedge_{r=1}^{k-1} df_{j+r \pmod{k}} \bigg] \left(X_{f_k}\right)                        \\
		 & \quad =\bigg[\sum_{j=0}^{k-1} \left(-1 \right)^{j (k-1)} \sum_{l=1}^{k-1} (-1)^l f_j X_{f_k}(f_{j+l \pmod{k+1}})  \bigwedge_{\substack{r=1       \\ r\neq l}}^{k-1} df_{j+r \pmod{k}} \bigg] \\
		 & \quad = \Reeb(f_k) \bigg[\sum_{j=0}^{k-1} \left(-1 \right)^{j (k-1)} \sum_{l=1}^{k-1} (-1)^{l+1} f_j f_{j+l \pmod{k+1}} \bigwedge_{\substack{r=1 \\ r\neq l}}^{k-1} df_{j+r \pmod{k}} \bigg] = 0,
	\end{align}
	where between the second and the third line we exploited the vanishing of the Jacobi bracket between each pair of functions $\{f_i\}$ and equation~\eqref{eqn:commnot0}.
	Thus, the contraction of the first expression in equation~\eqref{eqn:stepk-1} with $X_{f_k}$ is
	\begin{align}
		\label{eqn:firstindlineres}
		\left[\binom{n}{k} k! \right] \left(-1 \right)^{\frac12 k\left(k+1\right)} \left(df_k - \Reeb(f_k) \eta \right) \wedge \bigg[\sum_{j=0}^{k-1} \left(-1 \right)^{j (k-1)} f_j \bigwedge_{r=1}^{k-1} df_{j+r \pmod{k}} \bigg] \wedge d \eta^{\wedge (n-k)}.
	\end{align}

	By contracting the second expression in equation~\eqref{eqn:stepk-1} with the Hamiltonian vector field $X_{f_k}$, we obtain
	\begin{align}
		 & \left[\binom{n}{k} k! \right] \left(-1 \right)^{\frac12 k\left(k-1\right)} \left(\bigwedge_{j=0}^{k-1} df_j \wedge \eta \wedge d\eta^{\wedge n-k} \right) \left(X_{f_k}\right) \\
		 & \quad= \left[\binom{n}{k} k! \right] \left(-1 \right)^{\frac12 k\left(k-1\right)} \bigg( \sum_{l=0}^{k-1} (-1)^{l+1} f_l \Reeb(f_k) \bigwedge_{\substack{r=0                   \\ r\neq l}}^{k-1} df_{r} \wedge \eta \wedge d\eta^{\wedge (n-k)} \\
		 & \qquad\qquad + (-1)^{k+1} f_k \bigwedge_{j=0}^{k-1} df_j \wedge d\eta^{\wedge (n-k)}                                                                                           \\
		 & \qquad\qquad + (n-k) (df_k) \bigwedge_{j=0}^{k-1} df_j \wedge \eta \wedge d\eta^{\wedge (n-k-1)}\bigg),
	\end{align}
	in which we have used again the involutivity of the functions and equation~\eqref{eqn:commnot0}.
	Combining the two results we obtain
	\begin{align}
		 & \iota_{X_{f_0},X_{f_1}, \ldots, X_{f_k}}\left(\eta \wedge (d\eta)^{\wedge n} \right) =                                                                                                                                                            \\
		 & \quad = \left[\binom{n}{k} k! \right] \left(-1 \right)^{\frac12 k\left( k + 1 \right)} \left(df_k\right) \wedge \bigg[\sum_{j=0}^{k-1} \left(-1 \right)^{j (k-1)} f_j \bigwedge_{r=1}^{k-1} df_{j+r \pmod{k}} \bigg] \wedge d \eta^{\wedge (n-k)} \\
		 & \qquad + \left[\binom{n}{k} k! \right] \left(-1 \right)^{\frac12 k\left(k-1\right)} \bigg((-)^{k+1} f_k \bigwedge_{j=0}^{k-1} df_j \wedge d\eta^{\wedge (n-k)}                                                                                    \\
		 & \qquad\qquad +(n-k) (df_k) \bigwedge_{j=0}^{k-1} df_j \wedge \eta \wedge d\eta^{\wedge (n-k-1)}\bigg),
	\end{align}
	after simplifying the two terms involving $\Reeb(f_k)$.
	Rearranging the first two terms, we obtain
	\begin{align}
		 & \iota_{X_{f_0},X_{f_1}, \ldots, X_{f_k}} \left(\eta \wedge (d\eta)^{\wedge n} \right) =                                                                                                                               \\
		 & \quad = \left[\binom{n}{k} k! \right] \left(-1 \right)^{\frac12 (k+1) \left( k+2 \right)} \bigg[\sum_{j=0}^{k} \left(-1 \right)^{j k} f_j \bigwedge_{r=1}^{k} df_{j+r \pmod{k+1}} \bigg] \wedge d \eta^{\wedge (n-k)} \\
		 & \qquad + \left[\binom{n}{k+1} (k+1)! \right] \left(-1 \right)^{\frac12 k\left(k+1\right)} \bigwedge_{j=0}^{k} df_j \wedge \eta \wedge d\eta^{\wedge (n-k-1)}.
	\end{align}
	Applying the identity $\binom{n}{m} m! = \frac{n!}{(n-m)!}$ then concludes the proof.
\end{proof}

\section*{Acknowledgments}
The authors would like to thank Alessandro Bravetti, Angel Alejandro Garc\'ia-Chung \red{and Tom Mestdag for the useful discussions, and the anonymous referees for their valuable feedback.}

The research of MS and FZ was supported by the NWO project 613.009.10. The research of FZ was partially supported by FWO-EoS project Beyond symplectic geometry with UA Antigoon number 45816 and the grant Francqui Research Professor 2023-2026 of the Francqui Foundation with UA Antigoon number 49741.

\printbibliography

\appendix
\section{Tensor Densities}\label{app:tensor}

In this section, we outline the definition and some relevant properties of tensor densities. While densities used for integration ($\lambda=1$) are well-established, comprehensive overviews of general $\lambda$-densities are scarce. For a thorough and accessible recent review, we refer the reader to~\cite{Grandes_Umbert_2024}. To remain consistent with the projective geometry framework central to our work, we follow the terminology and approach established in~\cite{Fuks1986} and~\cite[Chapter 6.3]{projective}.

Let $M$ be a smooth manifold of dimension $m$.
The space of \emph{tensor densities of degree $\lambda \in\mathbb{R}$}, denoted by $F_\lambda$, is the space of sections of the bundle $\left(\bigwedge^m T^*M\right)^\lambda$, i.e. $\mathcal{F}_\lambda = \Gamma(\left(\bigwedge^m T^*M\right)^\lambda)$.
Locally in coordinates, any such generalized tensor $\phi$ can be expressed at a point $x \in M$ as
\begin{equation}
	\phi(x) = f(x)\ \lvert dx_1\wedge \cdots \wedge dx_{m}\rvert^\lambda,
\end{equation}
for some $\lambda\in\mathbb{R}$ and $f\in C^{\infty}(M)$.
A diffeomorphism $g\in \mathrm{Diff}(M)$ acts on a tensor densities $\phi \in \mathcal{F}_\lambda$ by
\begin{equation}
	\label{eqn:transf}
	g^{-1} \phi = \left(Dg\right)^\lambda\ \left(\phi \circ g\right),
\end{equation}
where $Dg$ is the Jacobian of the transformation $g$.

If $\zeta \in \Omega^m(M)$ is a volume form over $M$, a tensor density of degree $\lambda$ can be written in local coordinates as
\begin{equation}
	\phi(x) = f(x)\ \zeta^\lambda.
\end{equation}

If we consider a contact manifold $(M,\eta)$ of dimension {$m = 2n + 1$}, it is interesting to restrict the attention to its group of contactomorphisms~\cite[Chapter 7.5]{projective}.
Since a contactomorphism transforms the natural volume form $\eta \wedge(d\eta)^{\wedge n}$ as~\eqref{eqn:transfvolume} and the contact form $\eta$ as~\eqref{eqn:transfcontactform}, we can view the contact form as a tensor density of degree $\lambda=\frac{1}{n+1}$:
\begin{equation}
	\eta = \left(\eta \wedge (d\eta)^{\wedge n} \right)^{\frac{1}{n+1}}.
\end{equation}
We can then express any tensor density $\phi$ of degree $\lambda \in \mathbb{R}$ on a contact manifold as
\begin{equation}\label{eqn:app:tensorcontactexplaine}
	\phi(x) = f(x) \lvert\eta \wedge (d\eta)^{\wedge n}\rvert^{\lambda} = f(x)\, \eta^{\lambda(n+1)}.
\end{equation}

\section{Symplectic Hamitonian systems}\label{app:sympl}

Let $N$ be an even dimensional manifold and $\omega$ a non degenerate closed two-form, we call the couple $(N,\omega)$ a \emph{symplectic manifold}~\cite{Arnold2009,sympgeo} and $\omega$ a \emph{symplectic form}.

A real function $h\in C^1(N)$ is called an \emph{Hamiltonian function}, usually this function encodes the energy of a mechanical physical system.
Given an Hamiltonian function $h$, the symplectic forms allows to induce a vector field $V_h$ by
\begin{equation}
	\omega(V_h, \cdot) = d h,
\end{equation}
called (symplectic) Hamiltonian vector field.
We call the triple $(N,\omega, h)$ a symplectic Hamiltonian system to distinguish it from the contact Hamiltonian systems studied in the paper.
Since $\omega$ is skew-symmetric, one has
\begin{equation}\label{eq:hconsen}
	\omega(V_h,V_h) = V_h(h) = 0.
\end{equation}
Observing that $\dot H = \omega(V_h, V_h)$, intended as the time derivative along the flow of $V_h$, one can rephrase the identity above as conservation of energy.

\begin{defn}\label{def:liouville-vf}
	A vector field $\Xi$ on a symplectic manifold $(N, \omega)$ is a \emph{Liouville vector field} if it satisfies
	\begin{equation}
		L_{\Xi} \omega = \omega.
	\end{equation}
\end{defn}

\subsection*{Symplectization of a contact manifold}

Given a contact manifold $(M,\Delta)$ of dimension $2n+1$ (see Definition~\ref{def:contactmanifold}), one can naturally associate to it a symplectic manifold of dimension $2n+2$. This construction is called the \emph{symplectization} of $M$.

Recall that $\Delta\subset TM$ is a maximally non-integrable distribution. Consider the submanifold of the cotangent bundle
\begin{equation}
	SM := \bigl\{ (p,\alpha) \in T^*M \setminus \{0\} \;\big|\; \alpha|_{\Delta_p} = 0 \bigr\}.
\end{equation}
That is, $SM$ consists of all nonzero covectors at each point that annihilate the contact distribution. This is a smooth manifold of dimension $2n+2$, and it inherits from $T^*M$ the canonical Liouville 1-form $\theta$. We denote by
\[
	\lambda := \theta|_{SM}
\]
its restriction to $SM$. Then the 2-form $\omega := d\lambda$ is symplectic, making $(SM,\omega)$ the symplectization of $(M,\Delta)$.

If the contact structure is exact, with global 1-form $\eta$, then the symplectization admits a particularly simple global description:
\begin{equation}
	SM \;\cong\; M \times \mathbb{R}, \qquad
	\lambda = e^r \,\eta, \qquad
	\omega = d(e^r\eta) = e^r \left(dr \wedge \eta +  d\eta\right),
\end{equation}
where $r\in\mathbb{R}$ is the fiber coordinate. In this form, it is evident that $(SM,\omega)$ is an exact symplectic manifold with Liouville form $\lambda$.

Additionally, $SM$ carries a natural $\mathbb{R}$-action (scaling of covectors, or equivalently $r \mapsto e^t r$ in the trivialized model) generated by the Liouville vector field
\[
	W = \partial_r,
\]
which is, by construction, free and proper. This scaling action underlies the relation between the contact structure on $M$ and the symplectic geometry of its symplectization: The quotient of the symplectization $SM$ by the scaling action recovers the original contact manifold $(M,\Delta)$, where the contact distribution arises from the restriction of the symplectic form to a hypersurface transverse to $W$.

To each contact Hamiltonian vector field on $M$ we can associate a unique symplectic Hamiltonian vector field on $SM$. If $f \in C^{\infty}(M)$, the contact Hamiltonian vector field $X_f$, defined by~\eqref{eqn:Hamvec} is lifted to a symplectic Hamiltonian vector field $SX_f$ which satisfies:
\begin{equation}
	\pi_*(SX_f) := X_f,
\end{equation}
where $\pi: SM \to M$ is the natural projection $(m,r) \mapsto m$.

This can be used to lift some of the concepts developed in this manuscript and revisit them in the frame of symplectic geometry. See for instance Remark~\ref{rem:sympl}.

\end{document}